\documentclass{svjour3}

\usepackage{amssymb,amsmath,amsfonts,graphicx,tabularx}
\usepackage[compress]{cite}
\usepackage[usenames,dvipsnames]{color}\usepackage{verbatim}

\newcommand{\bra}[1]{\langle #1|}
\newcommand{\ket}[1]{|#1\rangle}
\newcommand{\braket}[2]{\langle #1|#2\rangle}
\newcommand{\ketbra}[2]{|#1 \rangle \langle #2|}
\def\E{{\cal E}}
\def\H{{\cal H}}
\def\I{{\cal I}}
\def\J{{\cal J}}
\def\IQ{\boldsymbol{\cal I}}
\def\JQ{\boldsymbol{\cal J}}
\def\Eset{\mathbb{E}}
\def\Rset{\mathbb{R}}

\def\rank{\operatorname{rank}}
\def\Tr{\operatorname{Tr}}

\def\Salicru{H_{(h,\phi)}}
\def\SalicruQ{\mathbf{H}_{(h,\phi)}}
\def\GenF{F_{(f,\alpha)}}

\def\GenFQ{\mathbf{F}_{(f,\alpha)}}
\def\id{\mathrm{id}}

\begin{document}


\title{A family of generalized quantum entropies: definition and properties}

\author{G.M. Bosyk \and S. Zozor \and F. Holik \and  M. Portesi \and P.W. Lamberti}
\institute{G.M. Bosyk  ( \email{gbosyk@fisica.unlp.edu.ar} ) \and  F. Holik \and
  M.  Portesi  \at   Instituto  de  F\'isica  La  Plata   (IFLP),  CONICET,  and
  Departamento de  F\'isica, Facultad de Ciencias  Exactas, Universidad Nacional
  de La Plata, C.C.~67, 1900 La  Plata, Argentina \and S.  Zozor \at Laboratoire
  Grenoblois d'Image,  Parole, Signal et  Automatique (GIPSA-Lab, CNRS),  11 rue
  des Math\'ematiques, 38402 Saint  Martin d'H\`eres, France \and P.W.  Lamberti
  \at  Facultad de  Matem\'atica, Astronom\'ia  y F\'isica  (FaMAF), Universidad
  Nacional  de  C\'ordoba,  and  CONICET,  Avenida Medina  Allende  S/N,  Ciudad
  Universitaria, X5000HUA, C\'ordoba, Argentina }

\maketitle


\begin{abstract}
  We  present  a  quantum   version  of  the  generalized  $(h,\phi)$-entropies,
  introduced by Salicr\'u \textit{et al.} for the study of classical probability
  distributions.   We establish their  basic properties,  and show  that already
  known quantum entropies such as  von Neumann, and quantum versions of R\'enyi,
  Tsallis, and  unified entropies, constitute particular classes  of the present
  general quantum Salicr\'u form.  We exhibit that majorization plays a key role
  in explaining  most of their common  features.  We give  a characterization of
  the quantum  $(h,\phi)$-entropies under the action of  quantum operations, and
  study  their properties  for composite  systems.  We  apply  these generalized
  entropies to the problem of detection of quantum entanglement, and introduce a
  discussion on possible generalized conditional entropies as well.
\end{abstract}




\section{Introduction}
\label{s:Introduction}

During the  last decades a vast field  of research has emerged,  centered on the
study  of  the  processing,  transmission  and storage  of  quantum  information
\cite{JozSch94,  Sch95,   NieChu10,  Ren13}.   In   this  field,  the   need  of
characterizing  and determining  quantum  states stimulated  the development  of
statistical  methods that  are suitable  for  their application  to the  quantum
realm~\cite{OgaHay04,  Hol11,  GilGut13,  YuDua14}.   This entails  the  use  of
entropic measures particularly  adapted for this task. For  this reason, quantum
versions of  many classical  entropic measures started  to play  an increasingly
important role, being von Neumann entropy~\cite{vNeu27} the most famous example,
with quantum versions of R\'enyi~\cite{Ren61} and Tsallis~\cite{Tsa88} entropies
as other widely known cases.  Many other examples of interest are also available
in the literature (see, for instance,~\cite{CanRos02, HuYe06, Kan02}).

Quantum entropic  measures are of use  in diverse areas of  active research. For
example, they find  applications as uncertainty measures (as is  the case in the
study  of  uncertainty  relations~\cite{MaaUff88, Uff90,  WehWin10,  ZozBos13,
  ZozBos14, ZhaZha15}); in  entanglement measuring and detection~\cite{HorHor94,
  AbeRaj01:01, TsaLlo01,  RosCas03, BenZyc06,  Hua13, OurHam15}; as  measures of
mutual  information~\cite{Yeu97, ZhaYeu98,  Car13,  GroWal13}; and  they are  of
great  importance  in the  theory  of  quantum  coding and  quantum  information
transmission~\cite{JozSch94, Sch95, WilDat13, DatRen13, AhlLob01}.

The alluded quantum entropies are nontrivially related, and while they have many
properties in common, they also  present important differences. In this context,
the study of generalizations of  entropic measures constitutes an important tool
for studying their  general properties.  In the theory  of classical information
measures,  Salicr\'u   entropies~\cite{SalMen93}  are,  up  to   now,  the  most
generalized extension  containing the Shannon~\cite{Sha48}, R\'enyi~\cite{Ren61}
and  Tsallis~\cite{Tsa88} entropies as  particular examples  and many  others as
well~\cite{HavCha67,  Dar70,  Rat91,  BurRao82}.    But  a  quantum  version  of
Salicr\'{u} entropies has not been studied yet in the literature.  We accomplish
this task by introducing a  natural quantum version of the classical expression.
Our  construction is  shown to  be of  great generality,  and contains  the most
important examples (von Neumann, and  quantum R\'enyi and Tsallis entropies, for
instance) as particular cases.

We  show that  several important  properties  of the  classical counterpart  are
preserved,   whereas  other  new   properties  are   specific  of   the  quantum
extension. In  our proofs, one of  the main properties  to be used is  the Schur
concavity,  which  plays  a  key  role,  in  connection  with  the  majorization
relation~\cite{BenZyc06, LiBus13} for (ordered) eigenvalues of density matrices.
Our generalization provides  a formal framework which allows  to explain why the
different quantum  entropic measures share  many properties, revealing  that the
majorization relation plays an important  role in their formal structure. At the
same time,  we give concrete  clues for the  explanation of the origin  of their
differences.   Furthermore,  the appropriate  quantum  extension of  generalized
entropies can be of use for defining information-theoretic measures suitable for
concrete  purposes. Given  our  generalized framework,  conditions  can then  be
imposed  in  order  to  obtain  families  of  measures  satisfying  the  desired
properties.

The paper is organized as follows. In Sec.~\ref{s:Review} we give a brief review
of (classical)  Salicr\'{u} entropies, also known  as $(h,\phi)$-entropic forms.
Our  proposal and  results are  presented in  Sec.~\ref{s:QuantumEntropies}.  In
\ref{s:Definition}   we    start   proposing   a   quantum    version   of   the
$(h,\phi)$-entropies  using a  natural trace  extension of  the  classical form,
followed    by    the   study    of    its    Schur-concavity   properties    in
\ref{s:SchurConcavity}. Then, in \ref{s:QEvolution}, we study further properties
related   to    quantum   operations   and   the    measurement   process.    In
\ref{s:CompositeSystemsI} we discuss the properties of quantum entropic measures
for   the  case   of  composite   systems  focusing   on  additivity,   sub  and
superadditivity properties,  whereas applications to  entanglement detection are
given  in \ref{s:CompositeSystemsII}.  Sec.  \ref{s:Furtherdefinitions} contains
an analysis of  informational quantities that could be  derived from the quantum
$(h,\phi)$-entropies.   Finally,  in   Sec.~\ref{s:Conclusions},  we  draw  some
concluding remarks.


\section{Brief review of classical $(h,\phi)$-entropies}
\label{s:Review}

Inspired   by   the  work   of   Csisz\'ar~\cite{Csi67},  Salicr\'u   \textit{et
  al.}~\cite{SalMen93} defined the $(h,\phi)$-entropies:
\begin{definition}
\label{def:Salicru}
Let us consider  an $N$-dimensional probability vector \ $p =  [p_1 \: \cdots \:
p_N]^t  \in  [0  , 1]^N$  \  with  \  $\sum_{i=1}^N  p_i  = 1$.   The  so-called
$(h,\phi)$-entropy is defined as
\begin{equation}
\label{eq:SalicruEnt}
\Salicru(p) = h\left( \sum_{i=1}^N \phi(p_i) \right),
\end{equation}
where  the \textit{entropic  functionals} \  $h: \Rset  \mapsto \Rset$  \  and \
$\phi: [0,1] \mapsto \Rset$ \ are  such that either: \ (i)~$h$ is increasing and
$\phi$  is concave, \  or \  (ii)~$h$ is  decreasing and  $\phi$ is  convex. The
entropic functional $\phi$ is assumed to be strictly concave/convex, whereas $h$
is taken to be strictly monotone,  together with $\phi(0) = 0$ and $h(\phi(1)) =
0$.
\end{definition}
We  notice   that  in   the  original  definition~\cite{SalMen93},   the  strict
concavity/convexity   and   monotony  characters   were   not  imposed.    These
considerations  will  allow  us  to  determine  the case  of  equality  in  some
inequalities presented  here.  The  assumption $\phi(0) =  0$ is natural  in the
sense   that  one   can  expect   the  elementary   information  brought   by  a
zero-probability event to be zero.  Also,  an appropriate shift in $h$ allows to
consider only  the case $h(\phi(1)) =  0$, thus not  affecting generality, while
giving  the vanishing of  entropy (i.e.,\  no information)  for a  situation with
certainty.

The $(h,\phi)$-entropies \eqref{eq:SalicruEnt}  provide a generalization of some
well-known   entropies   such    as   those   given   by   Shannon~\cite{Sha48},
R\'enyi~\cite{Ren61},  Havrda--Charv\'at,  Dar\'oczy or  Tsallis~\cite{HavCha67,
  Dar70, Tsa88}, unified  Rathie~\cite{Rat91} and Kaniadakis~\cite{Kan02}, among
many others.  In Table~\ref{tab:clasicalentropies}  we list some known entropies
and give the entropic functionals $h$  and $\phi$ that lead to these quantities.
Notice that  the entropies  given in  the table enter  in one  (or both)  of the
special families determined by entropic functionals  of the form: \ $h(x) = x$ \
and \  $\phi(x)$ concave~\cite{BurRao82}, or \ $h(x)  = \frac{f(x)}{1-\alpha}$ \
and   \   $\phi(x)   =   x^\alpha$~\cite{ZozBos14}.    Indeed,   the   so-called
$\phi$-entropy (or trace-form entropy) is defined as
\begin{equation}
\label{eq:phiEnt}
H_{(\id,\phi)}(p) = \sum_{i=1}^N \phi(p_i),
\end{equation}
where $\phi$ is concave with  $\phi(0) = 0$, whereas the $(f,\alpha)$-entropy is
defined as
\begin{equation}
\label{eq:fEnt}
\GenF(p) = \frac{1}{1-\alpha} \ f\left( \sum_{i=1}^N p_i^{\,\alpha} \right),
\end{equation}
where  $f$ is  increasing  with $f(1)=0$,  and  the \textit{entropic  parameter}
$\alpha$ is nonnegative and $\alpha \neq 1$. With the additional assumption that
$f$ is differentiable  and $f'(1) = 1$, one recovers the  Shannon entropy in the
limit $\alpha \to 1$.

\begin{table}[ht]
\begin{center}
\begin{tabular}
{
|>{\vspace{1mm}\centering}p{1.6cm}
|>{\vspace{1mm}\centering}p{4.5cm}
|>{\vspace{1mm}\hspace{2mm}}p{5cm}|
}
\hline
Name &
Entropic functionals &
\centerline{Entropy} \\[2mm]
\hline
Shannon &
$h(x) = x, \quad \phi(x) = - x \ln x\:$ &
$H(p) = - \sum_i p_i \ln p_i$ \\[2mm]
\hline
R\'enyi &
$ h(x) = \frac{\ln(x)}{1-\alpha} , \quad \phi(x) = x^\alpha$ &
$R_\alpha(p) = \frac{1}{1-\alpha}\ln\left( \sum_i p_i^{\,\alpha} \right)$\\[2mm]
\hline
%
Tsallis &
$h(x) = \frac{x-1}{1-\alpha} , \quad \phi(x) = x^\alpha $ &
$T_\alpha(p) = \frac{1}{1-\alpha} \left( \sum_i p_i^{\,\alpha} - 1 \right) $\\[2mm]
\hline
Unified &
$ h(x)= \frac{x^s-1}{(1-r)s} , \quad \phi(x) = x^r $ &
$E_r^s(p)= \frac{1}{(1-r)s} \left[ \left(\sum_i p_i^{\,r} \right)^s- 1 \right]$\\[2mm]
\hline
Kaniadakis    &
$ h(x)= x , \quad \phi(x) = \frac{x^{\kappa+1} - x^{-\kappa+1}}{2\kappa} $ &
$S_\kappa(p)= -\sum_i \frac{p_i^{\kappa+1} - p_i^{-\kappa+1}}{2\kappa}$\\[2mm]
\hline
\end{tabular}
\end{center}
\caption{Some well-known particular cases of $(h,\phi)$-entropies.}
\label{tab:clasicalentropies}
\end{table}

\

As  recalled  in Ref.~\cite{ZozBos14},  the  $(h,\phi)$-entropies share  several
properties as functions of the probability vector $p$:
\begin{itemize}
\item[$\bullet$] $\Salicru(p)$ is invariant  under permutation of the components
  of $p$.  Hereafter,  we assume that the components  of the probability vectors
  are written in decreasing order.
\item[$\bullet$] $\Salicru([p_1 \: \cdots \: p_N \: 0]^t) = H_{(h,\phi)}([p_1 \:
  \cdots \: p_N]^t)$: extending the space by adding zero-probability events does
  not change the value of the entropy (expansibility property).
\item[$\bullet$]  $\Salicru$  decreases  when  some events  (probabilities)  are
  merged,  that is, $\Salicru([p_1  \quad p_2  \quad p_3  \: \cdots  p_N]^t) \ge
  \Salicru([p_1+ p_2 \quad p_3 \: \cdots  p_N]^t)$. This is a consequence of the
  Petrovi\'c inequality that states that $\phi(a+b) \le \phi(a) + \phi(b)$ for a
  concave function $\phi$ vanishing at  0 (and the reverse inequality for convex
  $\phi$)~\cite[Th.~8.7.1]{Kuc09},     together      with     the     increasing
  (resp. decreasing) property of $h$.
\end{itemize}
Other    properties   relate    to    the   concept    of   majorization    (see
e.g.~\cite{MarOlk11}). Given two  probability vectors $p$ and $q$  of length $N$
whose components are  set in decreasing order, it is said  that $p$ is majorized
by $q$ (denoted  as $p \prec q$), when $\sum_{i=1}^n  p_i \leq \sum_{i=1}^n q_i$
for  all $n =  1, \ldots,  N-1$ and  $\sum_{i=1}^N p_i  = \sum_{i=1}^N  q_i$. By
convention, when  the vectors do not  have the same  dimensionality, the shorter
one is  considered to be  completed by zero  entries (notice that this  will not
affect  the  value  of   the  $(h,\phi)$-entropy  due  to  the  expansibility
property).  The majorization relation  allows to demonstrate some properties for
the $(h,\phi)$-entropies:
\begin{itemize}
\item[$\bullet$]  It is  strictly Schur-concave:  $p \prec  q \:  \Rightarrow \:
  \Salicru(p)  \ge \Salicru(q)$  with equality  if and  only if  $p =  q$.  This
  implies  that  the  more  concentrated  a  probability  vector  is,  the  less
  uncertainty  it represents  (or,  in  other words,  the  less information  the
  outcomes will bring).  The Schur-concavity of $\Salicru$ is consequence of the
  Karamata theorem~\cite{Kar32} that states that if $\phi$ is [strictly] concave
  (resp.\   convex),  then  $p   \mapsto  \sum_i   \phi  (p_i)$   is  [strictly]
  Schur-concave  (resp.\  Schur-convex) (see~\cite[Chap.~3,~Prop.~C.1]{MarOlk11}
  or~\cite[Th.~II.3.1]{Bha97}), together with  the [strictly] increasing (resp.\
  decreasing) property of $h$.
\item[$\bullet$] Reciprocally, if $\Salicru(p)  \ge \Salicru(q)$ for all pair of
  entropic  functionals $(h,\phi)$,  then $p  \prec  q$.  This  is an  immediate
  consequence  of Karamata theorem~\cite{Kar32}  (reciprocal part)  which states
  that  if  for   any  concave  (resp.\  convex)  function   $\phi$  one  has  \
  $\sum_{i=1}^n  \phi(p_i) \ge \sum_{i=1}^n  \phi(q_i)$, then  $\sum_{i=1}^n p_i
  \le \sum_{i=1}^n q_i$ for all $n  = 1,\ldots, N-1\:$ and $\:\sum_{i=1}^N p_i =
  \sum_{i=1}^N        q_i$       (see        also~\cite[A.3-(iv),~p.~14       or
  Ch.~4,~Prop.~B.1]{MarOlk11} or~\cite[Th.~II.3.1]{Bha97}).
\item[$\bullet$] It is bounded:
\begin{equation}
\label{eq:Entlowerupperbounded}
0 \le \Salicru(p) \le h \left( \|p\|_0 \, \phi \left( \frac{1}{\|p\|_0}
\right) \right) \le h \left( N \, \phi \left( \frac{1}{N} \right) \right) ,
\end{equation}
where $\|p\|_0$ stands  for the number of nonzero  components of the probability
vector.  The bounds are consequences  of the majorization relations valid to any
probability vector $p$ (see e.g.~\cite[p.~9,~Eqs.~(6)-(8)]{MarOlk11})
\[
\left[  \frac{1}{N}  \: \cdots  \:  \frac{1}{N}  \right]^t  \, \prec  \,  \left[
  \frac{1}{\|p\|_0} \: \cdots \: \frac{1}{\|p\|_0} \: 0 \: \cdots \: 0 \right]^t
\, \prec \, p \, \prec \left[ 1 \: 0 \: \cdots \: 0 \right]^t ,
\]
together with the Schur-concavity of $\Salicru$.  From the strict concavity, the
bounds  are  attained if  and  only if  the  inequalities  in the  corresponding
majorization relations reduce to equalities.
\end{itemize}

From   the    previous   discussion   we   can   see    immediately   that   the
$(h,\phi)$-entropies     fulfill    the     first     three    Shannon--Khinchin
axioms~\cite{Khi57},  which  are  (in   the  form  given  in  Ref.~\cite{Tem16})
(i)~continuity, (ii)~maximality (i.e.,\ it is maximum for the uniform probability
vector)  and  (iii)~expansibility.   The  fourth  Shannon--Khinchin  axiom,  the
so-called Shannon  additivity, is the rule  for composite systems  that is valid
only for  the Shannon  entropy (notice that  there are other  axiomatizations of
Shannon  entropy,  e.g.,~those given  by  Shannon  in~\cite{Sha48}  or by  Fadeev
in~\cite{Fad56}). A relaxation of Shannon additivity axiom, called composability
axiom, has been introduced~\cite{Tsa09,  Tem16}; it establishes that the entropy
of  a composite  system  should  be a  function  only of  the  entropies of  the
subsystems and a  set of parameters.  The class of  entropies that satisfy these
axioms (the first  three Shannon--Khinchin axioms and the  composability one) is
wide~\cite{Tem15}  but  nevertheless  can  be   viewed  as  a  subclass  of  the
$(h,\phi)$-entropies.

It has  recently been  shown that  the $(h,\phi)$-entropies can  be of  use, for
instance,  in  the study  of  entropic  formulations  of the  quantum  mechanics
uncertainty principle~\cite{ZozBos13, ZozBos14}.  They have also been applied in
the      entropic     formulation     of      noise--disturbance     uncertainty
relations~\cite{Ras16}.    Our  aim  is   to  extend   the  definition   of  the
$(h,\phi)$-entropies for  quantum density operators, and  study their properties
and potential applications in entanglement detection.


\section{Quantum $(h,\phi)$-entropies}
\label{s:QuantumEntropies}


\subsection{Definition and link with the classical entropies}
\label{s:Definition}

The von  Neumann entropy~\cite{vNeu27} can be  viewed as the  quantum version of
the classical Shannon entropy~\cite{Sha48},  by replacing the sum operation with
a  trace.   We  recall that  for  an  Hermitian  operator  $A  = \sum_i  a_i  \,
\ketbra{a_i}{a_i}$, with $\ket{a_i}$ being  its eigenvectors in $\H^N$ and $a_i$
being  the corresponding  eigenvalues, one  has $\phi(A)  = \sum_i  \phi(a_i) \,
\ketbra{a_i}{a_i}$,  and the  trace operation  is the  sum of  the corresponding
eigenvalues (i.e.,\  $ \Tr \phi(A) =  \sum_i \phi(a_i) $, where  $\Tr$ stands for
the trace operation).  In a  similar way to the classical generalized entropies,
we propose the following definition:
\begin{definition}
\label{def:QuantumSalicru}
Let us consider  a quantum system described by a  density operator $\rho$ acting
on an $N$-dimensional  Hilbert space $\H^N$, which is  Hermitian, positive (that
is,  $\rho  \ge  0$),  and  with   $\Tr  \rho  =  1$.   We  define  the  quantum
$(h,\phi)$-entropy as follows
\begin{equation}
\label{eq:QuantumSalicru}
\SalicruQ(\rho) = h\left( \Tr \, \phi(\rho) \right) ,
\end{equation}
where  the \textit{entropic  functionals} \  $h: \Rset  \mapsto \Rset$  \  and \
$\phi:  [0,1] \mapsto  \Rset$ \  are  such that  either: \  (i)~$h$ is  strictly
increasing  and  $\phi$  is  strictly  concave,  \ or  \  (ii)~$h$  is  strictly
decreasing and $\phi$ is strictly convex.   In addition, we impose $\phi(0) = 0$
and $h(\phi(1)) = 0$.
\end{definition}

The link between Eqs.~\eqref{eq:SalicruEnt} and~\eqref{eq:QuantumSalicru} is the
following.  Let  us consider  the  density  operator  written in  diagonal  form
(spectral    decomposition)    as   $\rho    =    \sum_{i=1}^N   \lambda_i    \,
\ketbra{e_i}{e_i}$,   \   with  eigenvalues   $\lambda_i   \geq  0$   satisfying
$\sum_{i=1}^N \lambda_i  = 1$, and being  $\{\ket{e_i}\}_{i=1}^N$ an orthonormal
basis. Then, the quantum $(h,\phi)$-entropy can be computed as
\begin{equation}
\label{eq:SalicruEntquantumdiagonal}
\SalicruQ(\rho) = \Salicru(\lambda) .
\end{equation}
This equation states  that the quantum $(h,\phi)$-entropy of  a density operator
$\rho$,  is nothing  but  the classical  $(h,\phi)$-entropy  of the  probability
vector $\lambda$  formed by the eigenvalues  of $\rho$. Notice  that despite the
link between  the quantal  and the  classical entropies defined  from a  pair of
entropic functionals $(h, \phi)$, we keep a different notation for the entropies
($\mathbf{H}$  and  $H$,  respectively)  in  order  to  distinguish  their  very
different meanings.

The  most relevant  examples of  quantum entropies,  which are  the  von Neumann
one~\cite{vNeu27},  quantum  versions  of  the  R\'enyi,  Tsallis,  unified  and
Kaniadakis  entropies~\cite{HuYe06,  Ras11:06,   FanCao15,  OurHam15},  and  the
quantum  entropies   proposed  in  Refs.~\cite{CanRos02,   Sha12},  are  clearly
particular cases of our quantum $(h,\phi)$-entropies~\eqref{eq:QuantumSalicru}.

In   what  follows,   we   give   some  general   properties   of  the   quantum
$(h,\phi)$-entropies (the validity of the properties for the von Neumann entropy
is  already  known,  see  for example~\cite{Lie75,  Weh78,  OhyPet93,  BenZyc06,
  NieChu10}).      In    our     derivations,    we     often     exploit    the
link~\eqref{eq:SalicruEntquantumdiagonal}.   With  that  purpose,  hereafter  we
consider, without loss of generality, that the eigenvalues of a density operator
$\rho$ are arranged in a (probability) vector $\lambda$, with components written
in decreasing order.


\subsection{Schur-concavity, concavity and bounds}
\label{s:SchurConcavity}

One of  the main  properties of the  classical $(h,\phi)$-entropies,  namely the
\emph{Schur-concavity}, is preserved in the quantum version of these entropies:
\begin{proposition}
\label{prop:Schurconcavity}
Let $\rho$ and $\rho'$ be two  density operators, acting on $\H^N$ and $\H^{N'}$
respectively, and such that $\rho \prec \rho'$. Then
\begin{equation}\label{eq:Schurconvavity}
\SalicruQ(\rho) \geq \SalicruQ(\rho'),
\end{equation}
with equality if and only if either $\rho'  = U \rho U^\dag$, or $\rho = U \rho'
U^\dag$, for any  isometric operator $U$ (i.e.,\ $U^\dag U  = I$), where $U^\dag$
stands for  the adjoint of $U$.   Reciprocally, if Eq.~\eqref{eq:Schurconvavity}
is satisfied for all pair of entropic functionals, then $\rho \prec \rho'$.
\end{proposition}

\begin{proof}
  Let  $\lambda$ and  $\lambda'$ be  the vectors  of eigenvalues  of  $\rho$ and
  $\rho'$, respectively, rearranged in decreasing order and adequately completed
  with zeros to  equate their lengths.  By definition,  $\rho \prec \rho'$ means
  that   $\lambda  \prec   \lambda'$  (see~\cite[p.~314,~Eq.~(12.9)]{BenZyc06}).
  Thus,  the   Schur-concavity  of  the  quantum   $(h,\phi)$-entropy  (and  the
  reciprocal  property) is inherited  from that  of the  corresponding classical
  $(h,\phi)$-entropy,  due   to  the  link~\eqref{eq:SalicruEntquantumdiagonal}.
  From  the  strict  concavity  or  convexity  of $\phi$  and  thus  the  strict
  Schur-concavity  of  the classical  $(h,\phi)$-entropies,  the equality  holds
  in~\eqref{eq:Schurconvavity}  if and  only if  $\lambda' =  \lambda$,  that is
  equivalent to have either $\rho' = U \rho U^\dag$ (when $N \le N'$) or $\rho =
  U \rho' U^\dag$ (when $N' \le N$).
\end{proof}

As  a direct  consequence, the  quantum  $(h,\phi)$-entropy is  lower and  upper
bounded, as in the classical case:
\begin{proposition}
\label{prop:boundedentropies}
The quantum $(h,\phi)$-entropy is lower and upper bounded
\begin{equation}\label{eq:boundedentropies}
0 \, \leq \, \SalicruQ(\rho) \, \leq \, h \left( \rank\rho \ \phi\left(
\frac{1}{\rank\rho} \right) \right) \, \leq \, h \left( N \, \phi\left(
\frac{1}{N} \right) \right),
\end{equation}
where  $\rank$  stands for  the  rank  of an  operator  (the  number of  nonzero
eigenvalues).   Moreover, the  lower bound  is  achieved only  for pure  states,
whereas the upper bounds are achieved for a density operator of the form $\rho =
\frac{1}{M}  \sum_{i=1}^M  \ketbra{e_i}{e_i}$   for  some  orthonormal  ensemble
$\{\ket{e_i}\}_{i=1}^M$, with $M = \rank\rho$ in the tightest situation and $M =
N$ in the  other one (in the latter case, necessarily  $\rho = \frac{1}{N} I_N$,
being $I_N$ the identity operator in $\H^N$).
\end{proposition}

\begin{proof}
  Let  $\lambda$ be  the vector  formed by  the eigenvalues  of  $\rho$. Clearly
  $\rank\rho = \|\lambda\|_0$, so that  the bounds are immediately obtained from
  that     of     the    classical     $(h,\phi)$-entropy,     due    to     the
  link~\eqref{eq:SalicruEntquantumdiagonal}.   Moreover, in the  classical case,
  $\Salicru(\lambda) = 0$ if and only if  $\lambda = [1 \: 0 \: \cdots \: 0]^t$,
  that is $\rho  = \ketbra{\Psi}{\Psi}$ is a pure state. On  the other hand, the
  upper   bounds   are   attained   if   and   only   if   $\lambda   =   \bigg[
  \underbrace{\frac{1}{M}  \: \cdots  \:  \frac{1}{M}}_M  \: 0  \:  \cdots \:  0
  \bigg]^t$, with $M = \rank \rho$ or $M = N$.
\end{proof}

\

The classical $(h,\phi)$-entropies and  their quantum versions are generally not
concave.  We establish here sufficient conditions on the entropic functional $h$
to ensure the concavity property of the quantum $(h,\phi)$-entropies.  We notice
that, with  the same  sufficient conditions, the  classical counterpart  is also
concave:
\begin{proposition}
\label{prop:concavity}
If the entropic  functional $h$ is concave, then  the quantum $(h,\phi)$-entropy
is concave, that is, for all $0\leq \omega \leq 1$,
\begin{equation}
\label{eq:concavity}
\SalicruQ(\omega\rho + (1-\omega)\rho') \geq \omega \,\SalicruQ(\rho) +
(1-\omega) \,\SalicruQ(\rho').
\end{equation}
\end{proposition}

\begin{proof}
  Let us  first recall the Peierls  inequality (see~\cite[p.~300]{BenZyc06}): if
  $\phi$ is  a convex function and  $\sigma$ is an Hermitian  operator acting on
  $\H^N$, then for any  arbitrary orthonormal basis $\{\ket{f_i}\}_{i=1}^N$, the
  following inequality holds
\begin{equation}
\label{eq:Peierlineq}
\Tr \phi(\sigma) \geq \sum_i \phi\left(\bra{f_i} \sigma \ket{f_i}\right).
\end{equation}
Consider  \  $\sigma  =  \omega  \rho  + (1-\omega)  \rho'  =  \sum_i  \lambda_i
\ketbra{e_i}{e_i}$ \  written in  its diagonal form,  $h$ decreasing  and $\phi$
convex. Then
\begin{equation*}\begin{array}{llll}
\Tr \phi(\sigma) & = & \displaystyle
\sum_i \phi(\lambda_i) =  \sum_i \phi(\bra{e_i} \sigma \ket{e_i})&\\[5mm]
& = & \displaystyle
\sum_i \phi(\bra{e_i} [\omega \rho + (1-\omega) \rho']  \ \ket{e_i})&\\[5mm]
& \le & \displaystyle \omega \sum_i \phi(\bra{e_i} \rho \ket{e_i}) + (1-\omega)
\sum_i \phi(\bra{e_i} \rho' \ket{e_i}) &
\: \mbox{[$\phi$ being convex]}\\[5mm]
& \le & \omega \Tr \phi(\rho) + (1-\omega) \Tr \phi(\rho')  &
\: \mbox{[due to Peierls inequality].}
\end{array}\end{equation*}
Notice that in the case $\phi$ concave, these two inequalities are reversed.
Thus, one finally has
\begin{eqnarray*}\begin{array}{llll}
h\left(\Tr \phi(\sigma)\right) & \ge & h\left(\omega \Tr \phi(\rho) + (1-\omega)
\Tr \phi(\rho')\right) &
\: \mbox{[$h$ being decreasing]} \\[5mm]
& \ge & \omega \, h( \Tr \phi(\rho)) + (1-\omega)  \ h( \Tr \phi(\rho'))
&
\: \mbox{[assuming $h$ concave].}
\end{array}\end{eqnarray*}
Notice that  in the  case $h$  increasing, the first  inequality holds,  and the
second inequality holds as well, from concavity of $h$ (with equality valid when
$h$ is the identity function).  Making use of Def.~\ref{def:QuantumSalicru}, the
proposition  is proved  in both  cases, under  the condition  that  the entropic
functional $h$ is concave.
\end{proof}

Note also that for the class  of $(f,\alpha)$-entropies, the concavity of $h$ is
equivalent to that  of $\frac{f}{1-\alpha}$.  Moreover, for the  von Neumann and
quantum Tsallis entropies the conditions of Proposition~\ref{prop:concavity} are
satisfied,  and  it  is well known  that  these  entropies have  the  concavity
property.   For quantum  R\'enyi  entropies, the  concavity  property holds  for
$0<\alpha  <1$  as  consequence  of  Proposition~\ref{prop:concavity},  but  for
$\alpha > 1$  the proposition does not apply  (see~\cite[p.~53]{BenZyc06} for an
analysis of concavity  in this range for classical  R\'enyi entropies).  For the
quantum  unified  entropies,  the  concavity  property holds  in  the  range  of
parameters   $r<1$  and   $s<1$   or   $r>1$  and   $s>1$   as  consequence   of
Proposition~\ref{prop:concavity},    which    complements    the    result    of
Ref.~\cite{HuYe06} and improves the result of Ref.~\cite{Ras11:06}.

It  is  interesting  to  remark  that  using the  concavity  property  given  in
Proposition~\ref{prop:concavity}, it is possible to define in a natural way, for
$h$  concave,  a  (Jensen-like)  quantum $(h,\phi)$-divergence  between  density
operators $\rho$ and $\rho'$, as follows:
\begin{equation}
\label{eq:Salicrudivergence}
\mathbf{J}_{(h,\phi)}\left(\rho,\rho'\right) =
\SalicruQ\left(\frac{\rho+\rho'}{2}\right) - \frac{1}{2} \left[\SalicruQ(\rho) +
\SalicruQ(\rho')\right],
\end{equation}
which is  nonnegative and symmetric  in its arguments.   This is similar  to the
construction  presented  in Ref.~\cite{BurRao82}  for  the  classical case,  and
offers  an  alternative   to  the  quantum  version  of   the  usual  Csisz\'ar
divergence~\cite{Csi67,  Sha12}.    It  can  be   shown  that  for   pure  sates
$\ket{\psi}$           and          $\ket{\psi'}$           the          quantum
$(h,\phi)$-divergence~\eqref{eq:Salicrudivergence} takes the form
\begin{equation}
\label{eq:Salicrudivergence_pure}
\mathbf{J}_{(h,\phi)}\left(\ketbra{\psi}{\psi},\ketbra{\psi'}{\psi'}\right) =
\Salicru\left( \left[\frac{1+|\braket{\psi}{\psi'}|}{2} \quad \frac{1-|
\braket{\psi}{\psi'}|}{2} \right]^t \right).
\end{equation}
Indeed, the  square root of  this quantity in  the von Neumann case  \ $[h(x)=x,
\phi(x)=-x\ln x]$ provides a metric for pure states~\cite{LamPor09}. Notice that
the  right-hand  side   of  Eq.~\eqref{eq:Salicrudivergence_pure}  is  a  binary
$(h,\phi)$-entropy.   Other basic  properties  and applications  of the  quantum
$(h,\phi)$-divergence are currently under study~\cite{BosBel16}.


\subsection{Specific properties of the quantum $(h,\phi)$-entropy }
\label{s:QEvolution}

We recall  that the quantum entropy  of a density operator  equals the classical
entropy of  the probability  vector formed by  its eigenvalues. In  other words,
considering  a density operator  as a  mixture of  orthonormal pure  states, its
quantum entropy coincides with the classical  entropy of the weights of the pure
states.  This  is not true  when the density  operator is not decomposed  in its
diagonal form, but  as a convex combination  of pure states that do  not form an
orthonormal basis.   The quantum $(h,\phi)$-entropy of  an arbitrary statistical
mixture of pure states, is  upper bounded by the classical $(h,\phi)$-entropy of
the probability vector formed by the mixture weights:
\begin{proposition}
\label{prop:staticalmixture}
Let  $\rho   =  \sum_{i=1}^M   p_i  \ketbra{\psi_i}{\psi_i}$  be   an  arbitrary
statistical mixture of pure  states $\ketbra{\psi_i}{\psi_i}$, with $p_i \geq 0$
and  $\sum_{i=1}^M  p_i=1$. Then,  the  quantum  $(h,\phi)$-entropy is  upper
bounded as
\begin{equation}
\label{eq:staticalmixture}
\SalicruQ(\rho) \leq \Salicru(p),
\end{equation}
where $p = [ p_1 \: \cdots \: p_M]^t$.
\end{proposition}

\begin{proof}
  First,      we       recall      the      the       Schr\"odinger      mixture
  theorem~\cite[Th.~8.2]{BenZyc06}:  a  density operator  in  its diagonal  form
  $\rho  =  \sum_{i=1}^N  \lambda_i  \ketbra{e_i}{e_i}$  can be  written  as  an
  arbitrary  statistical  mixture  of  pure  states  $\rho  =  \sum_{i=1}^M  p_i
  \ketbra{\psi_i}{\psi_i}$, with  $p_i \geq 0$ and $\sum_{i=1}^M  p_i=1$, if and
  only if, there exist a unitary $M \times M$ matrix $U$ such that
\begin{equation}\label{eq:mixture theorem}
\sqrt{p_i} \, \ket{\psi_i} = \sum_{j=1}^N U_{ij} \sqrt{\lambda_j} \, \ket{e_j}.
\end{equation}
As a corollary, one directly has~\cite{Nie00}
\begin{equation}
\label{eq:pBlambda}
p = B \lambda,
\end{equation}
where $B_{ij}=  |U_{ij}|^2$ are  the elements of  the $M \times  M$ bistochastic
matrix\footnote{It is  assumed that $M \ge  N$, otherwise $p$  is completed with
  zeros;  when  $M >  N$,  the  remaining $N-M$  terms  that  do  not appear  in
  Eq.~\eqref{eq:mixture theorem}  are added in  order to fulfill the  unitary of
  $U$  and $\lambda$  is to  be  understood as  completed with  zeros (for  more
  details,     see     the     proof     of    the     Schr\"odinger     mixture
  theorem~\cite[p.~222-223]{BenZyc06}).}   $B$.   From   the  lemma   of  Hardy,
Littlewood  and  P\'olya~\cite[Lemma~2.1]{BenZyc06} or~\cite[Th.~A.4]{MarOlk11},
this is equivalent  to the majorization relation $p  \prec \lambda$.  Therefore,
from~\eqref{eq:SalicruEntquantumdiagonal}   and  the   Schur-concavity   of  the
classical   $(h,\phi)$-entropy,   we   immediately   have   $\SalicruQ(\rho)   =
\Salicru(\lambda) \leq \Salicru(p)$.
\end{proof}

The previous proposition is a natural generalization of a well-known property of
von Neumann entropy.  One can also show that a related inequality holds:
\begin{proposition}
\label{prop:diagmap}
Let $\{\ket{e_k}  \}_{k=1}^N$ be an  arbitrary orthonormal basis of  $\H^N$ and,
for  a  given  density operator  $\rho$  acting  on  $\H^N$,  let us  denote  by
$p^E(\rho)$ the  probability vector with elements $p^E_k(\rho)  = \bra{e_k} \rho
\ket{e_k}$, that is, the diagonal elements of $\rho$ related to that basis. Then
\begin{equation}
\SalicruQ(\rho) \le \Salicru(p^E(\rho)).
\end{equation}
\end{proposition}

\begin{proof}
  The decomposing  of $\rho$ in  the basis $\{\ket{e_k}\}_{k=1}^N$ has  the form
  $\rho = \sum_{k,l=1}^N \rho_{k,l}  \ketbra{e_k}{e_l}$ where the diagonal terms
  are       $\rho_{k,k}       =       p_k^E(\rho)$.        The       Schur--Horn
  theorem~\cite[Th.~12.4]{BenZyc06}  states that the  vector $p^E(\rho)$  of the
  diagonal  terms  of  $\rho$  is  majorized  by the  vector  $\lambda$  of  the
  eigenvalues  of  $\rho$.   Thus,  from  the Schur-concavity  property  of  the
  classical  $(h,\phi)$-entropy, we  have  $\SalicruQ(\rho) =  \Salicru(\lambda)
  \leq \Salicru \left(p^E(\rho)\right)$.
\end{proof}

We consider now  the effects of transformations.  Among  them, unitary operators
are  important  since  the time  evolution  of  an  isolated quantum  system  is
described by  a unitary  transformation (i.e.,\ implemented  via the action  of a
unitary operator on the state).  One may expect that a ``good'' entropic measure
remains unchanged under such a  transformation. This property, known to be valid
for von  Neumann and quantum R\'enyi entropies~\cite{MulDup13}  among others, is
fulfilled for the quantum $(h,\phi)$-entropies,  and even in a slightly stronger
form,  i.e.,\  for isometries.   We  recall that  an  operator  $U: \H^N  \mapsto
\H^{N'}$  is said  to be  \emph{isometric}  if it  is norm  preserving. This  is
equivalent to $U^\dag U = I$. On the  other hand, an operator is then said to be
\emph{unitary} if it  is both isometric and co-isometric, that  is, both $U$ and
$U^\dag$ are isometric.  When $U: \H^N \mapsto \H^N$ (both Hilbert spaces having
the   same   dimension)  is   isometric,   it   is   necessarily  unitary   (see
e.g.~\cite{Hal82}).
\begin{proposition}
\label{prop:unitary}
The quantum  $(h,\phi)$-entropy is invariant under  any isometric transformation
$\rho \mapsto U \rho U^\dag$ where $U$ is an isometric operator:
\begin{equation}
\label{eq:unitary}
\SalicruQ(U \rho U^\dag) = \SalicruQ(\rho).
\end{equation}
\end{proposition}

\begin{proof}
  Let  us write  $\rho$ in  its diagonal  form, $\rho  =  \sum_{i=1}^N \lambda_i
  \ketbra{e_i}{e_i}$.   Clearly,   $U  \rho  U^\dag   =  \sum_{i=1}^N  \lambda_i
  \ketbra{f_i}{f_i}$,  where $\ket{f_i} =  U \ket{e_i}$  with $i=1,  \ldots, N$,
  form an  orthonormal basis (due  to the fact  that $U$ is an  isometry). Since
  $\rho$  and  $U  \rho U^\dag$  have  the  same  eigenvalues, and  thus,  using
  Eq.~\eqref{eq:SalicruEntquantumdiagonal}, we conclude  that they have the same
  $(h,\phi)$-entropy.
\end{proof}

\

When dealing with a quantum system, it  is of interest to estimate the impact of
a quantum operation  on it. In particular, one may guess  that a measurement can
only perturb the state and, thus,  that the entropy will increase.  This is also
true for  more general quantum operations.   Moreover, one may  be interested in
quantum entropies as signatures of an arrow of time: to this end one can see how
the value of  an entropic measure changes under the action  of a general quantum
operation.   More  concretely,  let   us  consider  general  quantum  operations
represented by completely positive  and trace-preserving maps $\E$, expressed in
the Kraus form  $\E(\rho) = \sum_{k=1}^K A_k \rho  A_k^\dag$ \ (with $\{A_k^\dag
A_k\}$ satisfying the completeness relation $\sum_{k=1}^K A_k^\dag A_k = I$). It
can be shown that the behavior  of entropic measures depends nontrivially on the
nature  of  the quantum  operation  (see e.g.~\cite[Sec.~12.6]{BenZyc06}).   For
example, a completely  positive map increases the von  Neumann entropy for every
state if  and only if  it is  bistochastic, i.e.,\ if  it is also  unital ($\{A_k
A_k^\dag\}$  also satisfies the  completeness relation),  so that  the operation
leaves the maximally mixed state invariant.  This is no longer true for the case
of a stochastic  (but not bistochastic) quantum operation.  What  can be said of
the  generalized  quantum  $(h,\phi)$-entropies?   This  is  summarized  in  the
following:
\begin{proposition}
\label{prop:bimap}
Let  $\E$ be  a  \emph{bistochastic  map}.  Then,  the  quantum operation  $\rho
\mapsto  \E(\rho)$  can  only  degrade  the  information  (i.e.,\  increase  the
$(h,\phi)$-entropy):
\begin{equation}
\label{eq:bimap}
\SalicruQ(\rho) \leq \SalicruQ(\E(\rho))
\end{equation}
with equality if and  only if $\E(\rho) = U \rho U^\dag$  for a unitary operator
$U$.
\end{proposition}

\begin{proof}
  From             the             quantum            Hardy--Littlewood--P\'olya
  theorem~\cite[Lemma.~12.1]{BenZyc06},  $\E(\rho)  \prec  \rho$,  so  that  the
  proposition  is   a  consequence  of   the  Schur-concavity  of   the  quantum
  $(h,\phi)$-entropy  (Proposition~\ref{prop:Schurconcavity}).  Let  us  mention
  that  an isometric  operator  can define  a  bistochastic map  only  if it  is
  unitary.
\end{proof}

This  is  a  well-known property  of  von  Neumann  entropy, when  dealing  with
projective measurements $A_k  A_l = \delta_{k,l} A_k$~\cite[Th.~11.9]{NieChu10}.
It turns out to  be true for the whole family of  $(h,\phi)$-entropies, and in a
more general context than projective  measurements.  However, as we have noticed
above, generalized  (but not bistochastic)  quantum operations can  decrease the
quantum    $(h,\phi)$-entropy.    Let    us   consider    the    example   given
in~\cite[Ex.~11.15,~p.~515]{NieChu10}.  Let $\rho$ be the density operator of an
arbitrary  qubit  system,  with  nonvanishing  quantum  $(h,\phi)$-entropy,  and
consider the generalized measurement performed by the measurement operators $A_1
= \ketbra{0}{0}$ and  $A_2 = \ketbra{0}{1}$ (a completely  positive map, but not
unital).  Then, the system after  this measurement is represented by \ $\E(\rho)
=  \ketbra{0}{0}  \rho  \ketbra{0}{0}   +  \ketbra{0}{1}  \rho  \ketbra{1}{0}  =
\ketbra{0}{0}$ with vanishing quantum $(h,\phi)$-entropy.

Note  that Proposition~\ref{prop:diagmap}  can  be viewed  as  a consequence  of
Proposition~\ref{prop:bimap}.  Indeed, it is straightforward to see that the set
of operators $E=\{\ketbra{e_k}{e_k}\}$ defines  the bistochastic map $\E(\rho) =
\sum_{k=1}^N           p^E_k(\rho)           \ketbra{e_k}{e_k}$.           Thus,
Proposition~\ref{prop:diagmap}    can   be    deduced    applying   successively
Proposition~\ref{prop:bimap} and Proposition~\ref{prop:staticalmixture}.

\

In  the  light of  the  previous discussions  and  results,  we can  reinterpret
Proposition~\ref{prop:diagmap}  as  follows:  the  quantum  $(h,  \phi)$-entropy
equals  the minimum  over the  set of  rank-one projective  measurements  of the
classical   $(h,   \phi)$-entropy   for   a  given   measurement   and   density
operator. Indeed, we  can extend the minimization domain to  the set of rank-one
positive operator valued measurements (POVMs)\footnote{  Recall that a POVM is a
  set $\{E_k\}$ of positive definite  operators satisfying the resolution of the
  identity}.  As  a consequence, we can  give an alternative  (and very natural,
from a physical perspective) definition for the $(h,\phi)$-entropies.
\begin{proposition}
\label{prop:povm1}
Let $\Eset$ be the set of all rank-one POVMs. Then
\begin{equation}
\SalicruQ(\rho) = \min_{E \in \Eset} \Salicru(p^E(\rho)),
\end{equation}
where $p^E(\rho)$ is the probability vector for the POVM $E = \{ E_k \}_{k=1}^K$
given the density operator $\rho$, i.e.,\ $p^E_k(\rho) = \Tr(E_k \rho)$.
\end{proposition}

\begin{proof}
  Let  us consider  an  arbitrary rank-one  POVM  $E =  \{  E_k \}_{k=1}^K$  and
  consider the positive operators $A_k = A_k^\dag = E_k^{\frac12}$.  Let us then
  define  $$  \E_E(\rho)  =  \sum_{k=1}^K  E_k^{\frac12}  \rho  E_k^{\frac12}  =
  \sum_{k=1}^K   p^E_k(\rho)  \,  \frac{E_k^{\frac12}   \rho  E_k^{\frac12}}{\Tr
    E_k^{\frac12}   \rho   E_k^{\frac12}  }   =   \sum_{k=1}^K  p^E_k(\rho)   \,
  \ketbra{\psi_k}{\psi_k},$$ where we have used the fact that $E_k$ is rank-one,
  so   its  square-root   can  be   written   in  the   form  $E_k^{\frac12}   =
  \ketbra{\tilde{e}_k}{\tilde{e}_k}$  (with $\ket{\tilde{e}_k}$  not necessarily
  normalized),  allowing  us  to  introduce  the  pure  states  $\ket{\psi_k}  =
  \frac{\ket{\tilde{e}_k}}{\braket{\tilde{e}_k}{\tilde{e}_k}^{\frac12}}$.    From
  the completeness relation satisfied by the POVM, $\E_E(\rho)$ is then a doubly
  stochastic map.  Thus,  applying successively Proposition~\ref{prop:bimap} and
  Proposition~\ref{prop:staticalmixture}   we   obtain   $$\SalicruQ(\rho)   \le
  \SalicruQ(\E(\rho))  \le \Salicru(p^E(\rho)).  $$ Since  $E$ is  arbitrary, we
  thus  have  $$\SalicruQ(\rho) \le  \min_{E  \in \Eset}  \Salicru(p^E(\rho)).$$
  Consider then $E_{\min} = \{ \ketbra{e_k}{e_k} \}_{k=1}^N$ where $\{ \ket{e_k}
  \}_{k=1}^N$   is    the   orthonormal   basis    that   diagonalizes   $\rho$.
  Thus $$\Salicru(p^{E_{\min}}(\rho))  = \SalicruQ(\rho) \le  \min_{E \in \Eset}
  \Salicru(p^E(\rho)) \le \Salicru(p^{E_{\min}}(\rho)),$$ which ends the proof.
\end{proof}

We notice that the alternative definition of quantum $(h,\phi)$-entropy given in
this proposition, can not be extended to any POVM.  The following counterexample
shows  this  impossibility.   Let  us  consider the  density  operator  $\rho  =
\frac{I_N}{N}$ with $N>2$ even, and the POVM $E =\left\{E_1,E_2 \right\}$ formed
by the positive  operators $E_1= \sum_{i=1}^{\frac{N}{2}} \ketbra{e_i}{e_i}$ and
$E_2       =       \sum_{i=\frac{N}{2}+1}^N      \ketbra{e_i}{e_i}$,       where
$\{\ket{e_i}\}_{i=1}^N$ is  an arbitrary orthonormal basis of  $\H^N$.  Thus, we
obtain  $p^E(\rho) = \left[  \frac12 \quad  \frac12 \right]^t$  and consequently
from   the   Schur-concavity   and    the   expansibility   of   the   classical
$(h,\phi)$-entropy  we have  $\SalicruQ(\rho)  = h\left(N  \phi\left(\frac{1}{N}
  \right)   \right)   >  h\left(2   \phi\left(\frac{1}{2}   \right)  \right)   =
\Salicru(p^E(\rho))$.


\subsection{Composite  systems  I: additivity,  sub  and superadditivities,  and
  bipartite pure states}
\label{s:CompositeSystemsI}

We  focus  now  on  some  properties of  the  quantum  $(h,\phi)$-entropies  for
bipartite  quantum systems  $AB$ represented  by density  operators acting  on a
product Hilbert space $\H_{AB}  = \H_A^{N_A} \otimes \H_B^{N_B}$.  Specifically,
we  are interested  in the  behavior  of the  entropy of  the composite  density
operator $\rho^{AB}$, with  reference to the entropies of  the density operators
of the subsystems\footnote{By definition,  the partial trace operation over $B$,
  \ $\Tr_B: \H_A^{N_A} \otimes \H_B^{N_B} \rightarrow \H_A^{N_A}$, is the unique
  linear operator such that \ $\Tr_B X_A  \otimes X_B = ( \Tr_B X_B)X_A$ for all
  $X_A$ and  $X_B$ acting on  $\H_A^{N_A}$ and $\H_B^{N_B}$,  respectively.  For
  instance,  let   us  consider  the   bases  $\{\ket{e_i^A}\}_{i=1}^{N_A}$  and
  $\{\ket{e_j^B}\}_{j=1}^{N_B}$  of $\H_A^{N_A}$ and  $\H_B^{N_B}$ respectively,
  and  the product  basis $\{\ket{e_i^A}  \otimes \ket{e_j^B}\}$  of $\H_A^{N_A}
  \otimes \H_B^{N_B}$.  Let us denote  by $\rho^{AB}_{i j,i' j'}$ the components
  in the product basis of  an operator $\rho^{AB}$ acting on $\H_A^{N_A} \otimes
  \H_B^{N_B}$.   Thus, the  partial  trace  over $B$  of  $\rho^{AB}$ gives  the
  density  operator of  the subsystem  $A$,  $\rho^A =  \Tr_B \rho^{AB}$,  whose
  components  are $\rho^A_{i,i'}  = \sum_j  \rho^{AB}_{i j,i'  j}$ in  the basis
  $\{\ket{e_i^A}\}$.},  $\rho^A   =  \Tr_B   \rho^{AB}$  and  $\rho^B   =  \Tr_A
\rho^{AB}$.

\

Now,  we give  sufficient  conditions  for the  additivity  property of  quantum
$(h,\phi)$-entropies:
\begin{proposition}
\label{prop:additivityH}
Let  $\rho^A \otimes  \rho^B$  be an  arbitrary  product density  operator of  a
composite  system $AB$,  and  $\rho^A$ and  $\rho^B$  the corresponding  density
operators of  the subsystems.  If, for $(a,b)  \in \left(0 \, ,  \, 1 \right]^2$
and    $(x,y)   \in   \left[\min\left\{\phi(1),N_A\phi\left(\frac{1}{N_A}\right)
  \right\}, \max\left\{\phi(1),N_A\phi\left(\frac{1}{N_A}\right) \right\}\right]
\times    \left[\min\left\{\phi(1),N_B\phi\left(\frac{1}{N_B}\right)   \right\},
  \max\left\{\phi(1),N_B\phi\left(\frac{1}{N_B}\right)  \right\}\right]$, $\phi$
and $h$ satisfy  the Cauchy functional equations either  of the form (i)~$\phi(a
b)  =  \phi(a) b+  a  \phi(b)$ and  $h(x+y)  =  h(x) +  h(y)$,  or  of the  form
(ii)~$\phi(ab)  =   \phi(a)  \phi(b)$  and   $h(xy)  =  h(x)+h(y)$.    Then  the
$(h,\phi)$-entropy satisfies the additivity property
\begin{equation}
\label{eq:additivityH}
\SalicruQ(\rho^A \otimes \rho^B) = \SalicruQ(\rho^A) + \SalicruQ(\rho^B)
\end{equation}
\end{proposition}

\begin{proof}
  In case~(i), by  writing the density operators $\rho^A$  and $\rho^B$ in their
  diagonal forms, it is straightforward to obtain
  \[
  \phi(\rho^A  \otimes \rho^B) =  \phi(\rho^A) \otimes  \rho^B +  \rho^A \otimes
  \phi(\rho^B),
  \]
  and thus
  \[
  h(\Tr \phi(\rho^A \otimes \rho^B)) =  h(\Tr \phi(\rho^A) + \Tr \phi(\rho^B)) =
  h(\Tr \phi(\rho^A)) + h(\Tr \phi(\rho^B))
  \]
  where we used $\Tr \rho^A = 1 = \Tr \rho^B$. Similarly, for case~(ii),
  \[
  \phi(\rho^A \otimes \rho^B) = \phi(\rho^A) \otimes \phi(\rho^B)
  \]
  and thus
  \[
  h(\Tr \phi(\rho^A \otimes \rho^B)) = h(\Tr \phi(\rho^A) \, \Tr \phi(\rho^B)) =
  h(\Tr \phi(\rho^A)) + h(\Tr \phi(\rho^B)).
  \]
  The  domains  where  the  functional   equations  have  to  be  satisfied  are
  respectively the  domain of definition of  $\phi$ and the image  of $\Tr \phi$
  (see Proposition~\ref{prop:boundedentropies}).
\end{proof}

Note that, on  the one hand, in case~(i) the functional  equation for $\phi$ can
be recast as $g(ab) = g(a) + g(b)$ with $g(x)=x^{-1} \phi(x)$.  Thus, $\phi(x) =
c_1 x \ln x$ and $h(x) = c_2 x$ with $c_1 c_2 < 0$ are entropic functionals that
are solutions  of the functional equations~(i)~\footnote{Notice  that the Cauchy
  equations $g(x+y) = g(x) + g(y)$,  $g(xy) = g(x)+g(y)$ and $g(xy) = g(x) g(y)$
  are  not necessarily linear,  logarithmic or  power type  respectively without
  additional assumptions on the domain where they are satisfied and on the class
  of admissible functions (see  e.g.~\cite{Cau21, Kuc09}).  But, recall that the
  entropic functionals $h$  and $\phi$ are continuous and  either increasing and
  concave,  or decreasing  and convex.}.   These solutions  lead to  von Neumann
entropy, which, as it is well known, is additive (see e.g.~\cite{Lie75, Weh78}).
On the other hand, in case~(ii), $\phi(x)  = x^\alpha$ and $h(x) = c \ln x$ with
$0 < \alpha < 1$ and $c > 0$ or $\alpha > 1$ and $c < 0$ are entropic functional
solutions. This is  the case for the R\'enyi entropies, which  are also known to
be    additive    (see     e.g.~\cite{MulDup13}).     In    general,    however,
$(h,\phi)$-entropies are  not additive,  for instance quantum  unified entropies
(including quantum Tsallis  entropies) do not satisfy this  property for all the
possible values of the entropic parameters~\cite{HuYe06, Ras11:06}. For the (not
so general) quantum $(f,\alpha)$-entropies, we can give necessary and sufficient
conditions for the additivity property:
\begin{proposition}
\label{prop:additivityF}
Let  $\rho^A \otimes  \rho^B$  be an  arbitrary  product density  operator of  a
composite  system $AB$,  and  $\rho^A$ and  $\rho^B$  the corresponding  density
operators of the subsystems.  Then, for any $\alpha > 0$ the additivity property
\begin{equation}
\label{eq:additivityF}
\GenFQ(\rho^A \otimes \rho^B) = \GenFQ(\rho^A) + \GenFQ(\rho^B)
\end{equation}
holds if  and only  if $f(xy) =  f(x) +  f(y)$ \ for  $(x,y) \in [  \min \{  1 ,
N_A^{1-\alpha} \} \, , \, \max \{1 ,  N_A^{1-\alpha} \} ] \, \times \, [ \min \{
1 , N_B^{1-\alpha} \} \, , \, \max \{1 , N_B^{1-\alpha} \} ]$.
\end{proposition}

\begin{proof}
  The `if'  part is  a direct consequence  of Proposition~\ref{prop:additivityH}
  where $\phi(x)  = x^{\alpha}$ and  $h(x) = \frac{f(x)}{1-\alpha}$  satisfy the
  Cauchy equations of condition~(ii).

  Reciprocally, if $\GenFQ$  is additive, we necessarily have  that $f\left( \Tr
    \left( \rho^A  \right)^\alpha \Tr \left(  \rho^B \right)^\alpha \right)  = f
  \left( \Tr \left(  \rho^A \right)^\alpha \right) + f  \left( \Tr \left( \rho^B
    \right)^\alpha \right)$ for any pair  of arbitrary states. Denoting $x = \Tr
  \left( \rho^A \right)^\alpha$  and $y = \Tr \left(  \rho^B \right)^\alpha$ and
  analyzing the  image of $\Tr \rho^\alpha$  for any density  operator acting on
  $\H^N$, we necessarily have $f(xy) = f(x) + f(y)$ over the domain specified in
  the proposition, which ends the proof.
\end{proof}

Notice  that,  if  $f$  is  twice  differentiable, one  can  show  that  $f$  is
proportional to  the logarithm  thus, among the  quantum $(f,\alpha)$-entropies,
only the von Neumann and quantum R\'enyi entropies are additive.

As we have seen, the $(h,\phi)$-entropies are, in general, nonadditive. However,
as suggested in~\cite{Rag95}, two types of subadditivity and superadditivity can
be of interest.  One of them compares the entropy of $\rho^{AB}$ with the sum of
the entropies of the subsystems $\rho^A$  and $\rho^B$ (global entropy vs sum of
marginal-entropies), and the other one  compares the entropy of $\rho^{AB}$ with
that  of  the   product  state  $\rho^A  \otimes  \rho^B$   (global  entropy  vs
product-of-marginals entropy).  The general  study of subadditivity of the first
type,  $\SalicruQ(\rho^{AB})  \leq  \SalicruQ(\rho^A) +  \SalicruQ(\rho^B)$,  is
difficult,  even if  one is  looking for  sufficient conditions  to  insure this
subadditivity.  Although it is not valid in general, there are certain cases for
which it holds.  For example, it holds for the von Neumann entropy~\cite{Weh78},
quantum  unified  entropies for  a  restricted  set of  parameters~\cite{HuYe06,
  Ras11:06},    and   quantum   Tsallis    entropy   with    parameter   greater
than~1~\cite{Rag95, Aud07}.  On the other hand, it is possible to show that only
the  von   Neumann  entropy  (or   an  increasing  function  of   it)  satisfies
subadditivity of the  second type, provided that some  smoothness conditions are
imposed on $\phi$.  This is summarized in the following:
\begin{proposition}
\label{prop:NoSubadditivity}
Let $\rho^{AB}$ be  a density operator of a composite  system $AB$, and $\rho^A$
and $\rho^B$ the corresponding density operators of the subsystems.  Assume that
$\phi$  is twice  differentiable  on  $(0 \,  ,  \,1)$.  The  $(h,\phi)$-entropy
satisfies
\begin{equation}
\label{eq:NoSubadditivity}
\SalicruQ(\rho^{AB}) \leq \SalicruQ(\rho^A \otimes \rho^B)
\end{equation}
if and only if $\SalicruQ$ is an increasing function of the von Neumann entropy,
given by $\phi(x) = - x \ln x$.
\end{proposition}

\begin{proof}
  The proof is based on two steps:
  \begin{enumerate}
  \item[$\bullet$] First, an example of a two qutrit diagonal system acting on a
    Hilbert space $\H^3  \otimes \H^3$ is presented, for which  it is shown that
    $\SalicruQ$ cannot  be subadditive, with the exception  of certain functions
    $\phi'$ satisfying a given functional equation.
  \item[$\bullet$]  Next,   under  the  assumptions  of   the  proposition,  the
    functional  equation  is solved,  and  it is  shown  that  all the  entropic
    functionals $\phi$ for  which we could not conclude  on the subadditivity of
    $\SalicruQ$,  can be  reduced to  the case  $\phi(x) =  - x  \ln x$  and $h$
    increasing.
  \end{enumerate}

  \

  {\bf Step 1}.  Consider the composite two qutrit systems acting on a Hilbert
  space $\H^3 \otimes \H^3$, of the form
  \[
  \rho^{AB} = \rho^A \otimes \rho^B  - c \big( \ketbra{00}{00} + \ketbra{11}{11}
  - \ketbra{10}{10} - \ketbra{01}{01} \big)
    \]
  with
  \[
  \rho^A  =  a  \, \ketbra{0}{0}  \,  +  \,  \alpha  \,  \ketbra{1}{1} \,  +  \,
  (1-a-\alpha)  \,  \ketbra{2}{2}  \quad   \mbox{and}  \quad  \rho^B  =  b  \,
  \ketbra{0}{0}    \,    +    \,    \beta    \,   \ketbra{1}{1}    \,    +    \,
  (1-b-\beta) \, \ketbra{2}{2}
  \]
  where $\{ \ket{0}, \ket{1}, \ket{2} \}$  is an orthonormal basis for $\H^3$, \
  $\ket{ij} =  \ket{i} \otimes \ket{j}$, the  coefficients $(a,\alpha,b,\beta) $
  in the set
  $$D = \{ a, \alpha, b, \beta: \:\: 0 < a,b < 1 \:\:
  \wedge \:\: 0 < \alpha \le 1-a \:\:  \wedge \:\: 0 < \beta \le 1-b \}$$
  and $c$ in the interval
  $$
  C_{a,\alpha,b,\beta} = \big[ - 1 + \max\big\{ a b , \alpha \beta , 1 - a \beta
  , 1 - \alpha b \big\} \, , \, \min\big\{ a b , \alpha \beta , 1 - a \beta , 1
  -\alpha   b   \big\}   \big].$$   Let   us  now   recall   the   Klein
  inequality~\cite[Eq.~(12.7)]{BenZyc06} for concave $\phi$,
  \[
  \Tr \phi(  \rho )  - \Tr \phi(  \sigma ) \:  \le \:  \Tr \left( \left(  \rho -
      \sigma \right) \, \phi'(\sigma) \right),
  \]
  the reversed inequality  holds for convex $\phi$.  If  the Klein inequality is
  applied to $\rho = \rho^A \otimes \rho^B$ and $\sigma = \rho^{AB}$, for $(a,b)
  \in ( 0  \, , \, 1)^2$ (such that $C_{a,\alpha,b,\beta}$  is not restricted to
  $\{0\}$),     and    $c     \in     \mathring{C}_{a,\alpha,b,\beta}$    (where
  $\mathring{\cdot}$  denotes the  interior of  a  set), we  obtain for  concave
  $\phi$,
  \begin{equation}
  \Tr\phi(\rho^A \otimes \rho^B) - \Tr\phi(\rho^{AB}) \: \le \: c \:
   g(a,\alpha,b,\beta,c),
  \label{eq:KleinG}
  \end{equation}
  and the reversed inequality for convex $\phi$, where
  \begin{equation}
  g(a , \alpha , b , \beta , c) = \phi'\big( a b - c \big) + \phi'\big( \alpha
  \beta - c \big) - \phi'\big( a \beta + c \big) - \phi'\big( \alpha b + c \big).
  \end{equation}
  Assume that there exists  $(x,u,y,v) \in \mathring{D}$ such that $g(x,u,y,v,0)
  \ne 0$.  From the continuity of  $\phi'$, function $g$ is continuous, and thus
  there exists  a neighborhood $V_0 \subset \mathring{C}_{x,u,y,v}$  of $0$ such
  that function  $c \mapsto g(x,u,y,v,c)$  has a constant  sign on $V_0$.   As a
  conclusion,  $c  \mapsto  c \,  g(x,u,y,v,c)$  does  not  preserve sign  on  $
  V_0$. This allows us to conclude from~\eqref{eq:KleinG} that when $\phi$
  is concave (resp.  convex), $\Tr\phi(\rho^{AB})$ can be higher (resp.  lower) than
  $\Tr\phi(\rho^A   \otimes   \rho^B)$.   Together   with   the  increasing   (resp.
  decreasing) property of $h$, it  is then clear that if $g(a,\alpha,b,\beta,0)$
  is not identically zero on the domain $\mathring{D}$, then $\SalicruQ$
  cannot be subadditive in the sense global vs product of marginals.

  \

  {\bf Step 2}.  If $g(a,\alpha,b,\beta,0)  = 0$ on $\mathring{D}$, then $\phi'$
  satisfies the functional equation
  \begin{equation}
  \phi'\big( a b \big) + \phi'\big( \alpha \beta \big) -
  \phi'\big( a \beta \big) - \phi'\big( \alpha b \big) = 0,
  \label{eq:FunctionalEquation}
  \end{equation}
  and  one  cannot  use  the  previous  argument to  decide  if  $\SalicruQ$  is
  subadditive    or    not.    In    order    to    solve    this   riddle    we
  follow~\cite[\S~6]{DarJar79},   where  a   similar   functional  equation   is
  discussed.   By  fixing   $(a,b)  \in  (0  \,  ,   \,  1)^2$,  differentiating
  identity~\eqref{eq:FunctionalEquation}   with    respect   to   $\alpha$   and
  multiplying the result by $\alpha$, we obtain
  \[
  \alpha \beta  \, \phi''(\alpha  \beta) = \alpha  b \, \phi''(\alpha  b) \qquad
  \mbox{for} \qquad (\alpha,\beta) \in (0 \, , \, 1-a) \times (0 \, , \, 1-b).
  \]
  This means  that $x  \, \phi''(x)$ is  constant for  $x \in (0  \, ,  \, (1-a)
  \max\{b,1-b\})$, for all $(a,b) \in (0  \, , \, 1)^2$.  Thus, $x \, \phi''(x)$
  is constant for $x \in (0 \, , \, 1 )$.  In other words, $\phi$ is necessarily
  of the  form $\phi(x)  = -  \lambda \,  x \ln x  + \mu  x +  \nu$. Due  to the
  continuity of $\phi$, this is valid on  the closed set $[0 \, , \, 1]$.  Since
  $\Tr \rho = 1$, one can restrict the  analysis to $\mu = 0$ (this value can be
  put in  $\nu$, leaving the  entropy unchanged).  Moreover, this  constant does
  not alter  the concavity or  convexity of  $\phi$ and thus  can be put  in $h$
  [without altering  its monotonicity and,  thus, the sense of  the inequalities
  between $\SalicruQ(\rho^{AB})$ and $\SalicruQ(\rho^A \otimes \rho^B)$ either].
  To ensure strict concavity (convexity) of  $\phi$, one must have $\lambda > 0$
  (resp.  $\lambda < 0$) and thus,  without loss of generality, $\lambda$ can be
  rejected in $h$. Finally,  one can rapidly see that $\phi (x) =  - x \, \ln x$
  satisfies the identity~\eqref{eq:FunctionalEquation}.

  \

  As a conclusion, under the assumptions of the proposition, when $\SalicruQ$ is
  not  an  increasing  function  of  the  von  Neumann  entropy,  it  cannot  be
  subadditive.  Reciprocally, the von Neumann entropy is known to be subadditive
  (see  e.g.~\cite{Lie75, Weh78}),  and  this remains  valid  for any  increasing
  function of this entropy, which finishes the proof.
\end{proof}

Notice that neither R\'enyi nor Tsallis entropies satisfy this subadditivity for
any entropic parameter  except for $\alpha = 1$ \footnote{For  $\alpha = 0$ this
  subadditivity is also satisfied, but note that in this special case, $\phi$ is
  not  continuous  and   moreover  does  not  fulfill  the   conditions  of  the
  proposition.}   (i.e.,\   von  Neumann   entropy).   As  a   consequence,  from
Propositions~\ref{prop:additivityF}  and~\ref{prop:NoSubadditivity},  we  obtain
that, except for the von  Neumann case (and the zero-parameter entropy), R\'enyi
entropies   do   not  neither   satisfy   usual   subadditivity   in  terms   of
$\SalicruQ(\rho^{AB})$  and $\SalicruQ(\rho^A) +  \SalicruQ(\rho^B)$ (as  in the
classical    counterpart~\cite[Ch.~IX,~\S6]{Ren70}).      In    addition,    the
counterexample used in the proof of  the proposition, allows us to conclude that
the  same   nonsubadditivity  also  holds  for  the   classical  counterpart  of
$(h,\phi)$-entropies.

\

Regarding both types  of superadditivity, it is well known  that the von Neumann
entropy does  not satisfy  neither of them.   Here, we  extend this fact  to any
$(h,\phi)$-entropy, as summarized in the following:
\begin{proposition}
\label{prop:NoSuperadditivity}
Let $\rho^{AB}$ be  a density operator of a composite  system $AB$, and $\rho^A$
and  $\rho^B$  the  corresponding  density  operators  of  the  subsystems.  The
$(h,\phi)$-entropy is nonsuperadditive in the sense that
\begin{eqnarray}
\SalicruQ(\rho^{AB}) & \ge & \SalicruQ(\rho^A \otimes \rho^B) \quad \mbox{and}\\
\SalicruQ(\rho^{AB}) & \ge & \SalicruQ(\rho^A) + \SalicruQ(\rho^B)
\end{eqnarray}
are not satisfied for all states.
\end{proposition}

\begin{proof}
  Let us consider  the two qubit diagonal system acting on $\H^2 \otimes \H^2$:
  \[
  \rho^{AB} = \frac{1}{2}  (\ketbra{00}{00}  + \ketbra{11}{11}),
  \]
  which  gives  $\rho^A  = \frac{I_2}{2}=  \rho^B  $.   In  this case,  we  have
  $\SalicruQ(\rho^{AB}) = h\left( 2 \phi(\frac12) \right)$, $\SalicruQ(\rho^A) +
  \SalicruQ(\rho^B) = 2\,  h\left( 2 \phi(\frac12)\right)$ and $\SalicruQ(\rho^A
  \otimes   \rho^B)   =   h\left(    4   \phi(\frac14)   \right)$,   such   that
  $\SalicruQ(\rho^{AB})  < \SalicruQ(\rho^A)  + \SalicruQ(\rho^B)$  (due  to the
  positivity  of  the entropies)  and  $\SalicruQ(\rho^{AB}) <  \SalicruQ(\rho^A
  \otimes \rho^B)$ (due to the Schur-concavity property).
\end{proof}

\

For  the case  of von  Neumann entropy,  it is  well known  that the  entropy of
subsystems     of    a     bipartite     pure    state     are    equal     (see
e.g.~\cite[Th.~11.8-(3)]{NieChu10},).   We  extend this  result  to any  quantum
$(h,\phi)$-entropy.
\begin{proposition}
\label{prop:pureAB}
Let $\ket{\psi}$ be a pure state of  a composite system $AB$ and $\rho^A = \Tr_B
\ketbra{\psi}{\psi}$ and $\rho^B  = \Tr_A \ketbra{\psi}{\psi}$ the corresponding
density operators of the subsystems. Then
\begin{equation}
\label{eq:pureAB}
\SalicruQ(\rho^A) = \SalicruQ(\rho^B).
\end{equation}
\end{proposition}

\begin{proof}
  From  the  Schmidt  decomposition theorem~\cite[Th.~9.1]{BenZyc06},  any  pure
  state $\ket{\psi} \in \H_A^{N_A} \otimes  \H_B^{N_B}$ can be written under the
  form
\begin{equation}
\label{eq:Schmidt}
\ket{\psi} = \sum_{i=1}^N \sqrt{\lambda_i} \, \ket{e_i^A} \otimes \ket{e_i^B},
\end{equation}
where  $\{\ket{e_i^A}\}_{i=1}^{N_A}$ and  $\{\ket{e_i^B}\}_{i=1}^{N_B}$  are two
orthonormal bases  for $\H_A^{N_A}$ and  $\H_B^{N_B}$, respectively, and  $N \le
\min\left\{N_A,N_B \right\}$.  The density operators of the subsystems are then
\begin{equation}
\rho^A = \sum_{i=1}^{N_A} \lambda_i \ketbra{e_i^A}{e_i^A} \qquad \mbox{and} \qquad
\rho^B = \sum_{i=1}^{N_B} \lambda_i \ketbra{e_i^B}{e_i^B},
\end{equation}
so  that the first  $N$ eigenvalues  $\lambda_i$ are  equal, the  remaining ones
being zero.   Therefore, using  the expansibility property,  both have  the same
quantum $(h,\phi)$-entropy.
\end{proof}


\subsection{Composite systems II: entanglement detection}
\label{s:CompositeSystemsII}

Now, we  consider the  use of quantum  $(h,\phi)$-entropies in  the entanglement
detection problem.  As  with the classical entropies, one  would expect that the
quantum entropies  of density  operators reduced to  subsystems were  lower than
that of  the density operator  of the composite  system. We show here  that this
property turns out to be valid for separable density operators. We recall that a
bipartite  quantum  state  is  separable  if  it can  be  written  as  a  convex
combination of  product states~\cite{Wer89}, that  is\footnote{Equivalently, the
  pure  states  $\ketbra{\psi_m^A}{\psi_m^A}$ and  $\ketbra{\psi_m^B}{\psi_m^B}$
  can   be   replaced  by   mixed   states   defined   on  $\H^A$   and   $\H^B$
  respectively~\cite{NieKem01}.},
\begin{equation}\label{eq:separable}
\rho^{AB}_{\,\mathrm{Sep}} = \sum_{m=1}^M \omega_m \,\ketbra{\psi_m^A}{\psi_m^A}
\otimes \ketbra{\psi_m^B}{\psi_m^B} \qquad \mbox{with} \qquad \omega_m \geq 0 \
\mbox{and} \ \sum_{m=1}^M \omega_m=1.
\end{equation}

For bipartite separable states, we have the following:
\begin{proposition}
\label{prop:entcriteria}
Let $\rho^{AB}_{\,\mathrm{Sep}}$ be a  separable density operator of a composite
system  $AB$,  and  let  $\rho^A$  and $\rho^B$  be  the  corresponding  density
operators of the subsystems. Then
\begin{equation}
\label{eq:entcriteria}
\SalicruQ(\rho^{AB}_{\,\mathrm{Sep}}) \geq \max \left\{ \SalicruQ(\rho^A) \, , \,
\SalicruQ(\rho^B) \right\} ,
\end{equation}
for any pair of entropic functionals $(h, \phi)$.
\end{proposition}

\begin{proof}
  This  is a  corollary of  a more  general criterion  of separability  given in
  Ref.~\cite{NieKem01}  (also given  in~\cite[B.4,~p.~386]{BenZyc06}),  based on
  majorization.  Indeed,  from  that  criterion, a  separable  density  operator
  $\rho^{AB}_{\,\mathrm{Sep}}$  and the reduced  density operators  $\rho^A$ and
  $\rho^B$ satisfy the majorization relations
  \begin{equation}\label{eq:majorization separable}
  \rho^{AB}_{\,\mathrm{Sep}} \prec \rho^A \qquad \mbox{and} \qquad
  \rho^{AB}_{\,\mathrm{Sep}} \prec \rho^B.
  \end{equation}
  Inequality~\eqref{eq:entcriteria} is thus a consequence of the Schur-concavity
  of $\SalicruQ$, proved in Proposition~\ref{prop:Schurconcavity}.
\end{proof}

It is  worth mentioning  that the majorization  relations \eqref{eq:majorization
  separable}    do    not    imply    the   separability    of    the    density
operator~\cite{BenZyc06, NieKem01} and thus \eqref{eq:majorization separable} is
a sufficient  condition for  the derivation of~\eqref{eq:entcriteria}.  In other
words, some pair(s) of entropic functionals $(h,\phi)$ and a density operator of
the composite system violate~\eqref{eq:entcriteria}.

Proposition~\ref{prop:entcriteria}  was   proved  originally  for   von  Neumann
entropy~\cite{HorHor94}, and later  on for some other quantum  entropies such as
the R\'enyi, Tsallis, and Kaniadakis ones (see e.g.~\cite{AbeRaj01:01, TsaLlo01,
  CanRos02,  RosCas03, OurHam15}  or~\cite[B.5,~p.~387]{BenZyc06}).  Remarkably,
this property turns out to be fulfilled by any quantum $(h,\phi)$-entropy.

\

As  an  example  we   use  Proposition~\ref{prop:entcriteria}  in  the  case  of
$(f,\alpha)$-entropies, in  order to verify  its efficiency to  detect entangled
Werner  states  of  two qubit  systems.  Werner  density  operators are  of  the
form~\cite{Wer89} or~\cite[Eq.~(15.42),~p.~382-383]{BenZyc06}:
\begin{equation}
\label{eq:Wernerstates}
\rho^{AB} = \omega \,\ketbra{\Psi^-}{\Psi^-} + (1-\omega) \frac{I_4}{4},
\end{equation}
where  $\ket{\Psi^-}=\frac{1}{\sqrt{2}}   (\ket{00}-\ket{11})$  is  the  singlet
state, $\ket{0}$ and  $\ket{1}$ are eigenstates of the  Pauli matrix $\sigma_z$,
and $\omega  \in [0,1]$.  It is well  known that Werner states  are entangled if
and only  if $\omega >  \frac13$.  The density  operators of the  subsystems are
$\rho^A     =     \frac{I_2}{2}     =     \rho^B$.      Therefore,     following
Proposition~\ref{prop:entcriteria}  for an  $(f,\alpha)$-entropy, we  can assert
that the Werner states are entangled if the function
\begin{equation}
\label{eq:ZWernerstates}
Z_{f,\alpha}(\omega) = \left\{\begin{array}{lll} \displaystyle \frac{ f
\left( 3 \left(\frac{1-\omega}{4}\right)^\alpha  +
\left(\frac{1+3\omega}{4}\right)^\alpha \right) - f \left(2^{1-\alpha}
\right)}{\alpha-1} , & \qquad & \alpha \ne 1\\[5mm]
\displaystyle \frac{3 \, (1 - \omega) \ln(1-\omega) \, + \, (1+3\omega) \ln
(1+3\omega)}{4} - \ln 2 , & \qquad & \alpha = 1
\end{array}\right.
\end{equation}
is positive.   Note that,  since $f$ is  increasing, the sign  of $Z_{f,\alpha}$
does not depend on the choice of $f$, so that we can take $f(x) = \ln x$ without
loss   of  generality.    Figure~\ref{fig:wernerent}  is   a  contour   plot  of
$Z_{\ln,\alpha}(\omega)$  versus   $\omega$  and  $\alpha$.    The  dashed  line
represents the boundary between the entangled situation ($\omega > \frac13$) and
the separable  one ($\omega  < \frac13$), and  the solid line  distinguishes the
situation  $Z_{f,\alpha}(\omega)  >  0$   (to  the  right)  from  the  situation
$Z_{f,\alpha}(\omega) < 0$ (to the left).  It can be seen that, in this specific
example,  the   entropic  entanglement  criterion  is   improved  when  $\alpha$
increases.  This can be well understood noting that, when $\alpha \to \infty$, \
$Z_{\ln,\alpha}(\omega) \to  \ln \left(  \frac{1+3 \omega}{2} \right)$,  that is
positive if  and only if  $\omega >  \frac13$, i.e.,\ if  and only if  the Werner
states are entangled.

\begin{figure}[htbp]
\begin{center}
\includegraphics{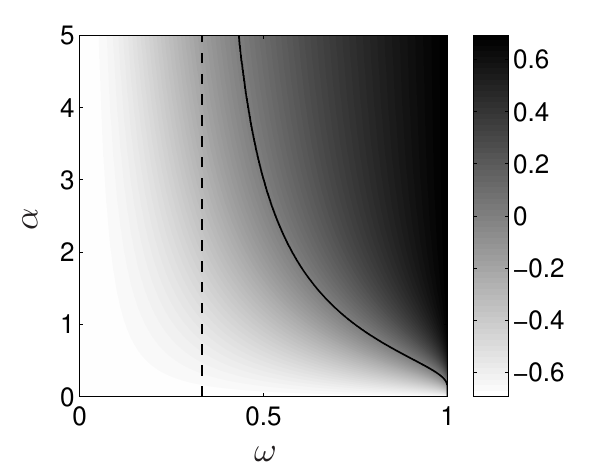}
\end{center}
\caption{Contour plot  of function $Z_{\ln,\alpha}(\omega)$  versus $\omega$ and
  $\alpha$,  as given  in  Eq.~\eqref{eq:ZWernerstates}.  To  the  right of  the
  dashed line  (at $\omega=\frac 13$)  the Werner states~\eqref{eq:Wernerstates}
  are  entangled,  while  to  its  left  they are  separable.   The  solid  line
  corresponds to  $Z_{f,\alpha}(\omega) = 0$  and limits two situations:  to the
  right  where $Z_{f,\alpha}(\omega)$ is  positive and  the criterion  allows to
  conclude that the  states are entangled, and to the left  where nothing can be
  said about the states.}
\label{fig:wernerent}
\end{figure}

This  simple  illustration  aims  at  showing  that  the  use  of  a  family  of
$(h,\phi)$-entropies  instead   of  a  particular  one,  or   playing  with  the
parameter(s)  of  parametrized   $(h,\phi)$-entropies,  allows  one  to  improve
entanglement detection.

Naturally the majorization entanglement  criterion is stronger than the entropic
one. Indeed, for the example given above, the majorization criterion detects all
entangled Werner  states.  However,  there are situations  where the  problem of
computation of the eigenvalues of the density operator happens to be harder than
the calculation  of the trace in  the entropy definition (at  least for entropic
functionals of the  form $\phi(x) = x^n$, with $n$  integer). Moreover, from the
converse  of  Karamata   theorem  (see  Sec.~\ref{s:Review}),  the  majorization
criterion  becomes   equivalent  to  the  entropic  one   when  considering  the
\textit{whole} family  of $(h,\phi)$-entropies.  This  allows us to  expect that
the  more ``nonequivalent''  entropies  are used,  the  better the  entanglement
detection should be.

\

Another motivation for the use of general entropies in entanglement detection is
that, in a more realistic scenario, one does not have complete information about
the  density  operator, so  the  majorization  criterion,  and consequently  the
entropic  one can  not be  applied.   It happens  usually that  one has  partial
information from mean values of certain  observables. In that case, one needs to
use some inference  method to estimate the density  operator compatible with the
available information.  One  of the more common methods  for obtaining the least
biased density operator  compatible with the actual information,  is the maximum
entropy principle~\cite{Jay57} (MaxEnt for  brevity).  That is, the maximization
of  von  Neumann entropy  subject  to the  restrictions  given  by the  observed
data. However, this  procedure can fail when dealing  with composite systems, as
shown  in Ref.~\cite{HorHor99}.  Indeed,  MaxEnt using  von Neumann  entropy can
lead to fake  entanglement, which means that it  predicts entanglement even when
there   exists    a   separable   state   compatible   with    the   data.    In
Ref.~\cite{CanRos02} it was  shown that, using concave quantum  entropies of the
form $\mathbf{H}_{(\id,\phi)}(\rho) = \Tr (\phi(\rho))$, it is possible to avoid
fake  entanglement   when  the  partial   information  is  given   through  Bell
constraints.

\

Now  we address the  following question  that arises  in a  natural way.   Is it
possible  to  use   the  constructions  given  above  to   say  something  about
multipartite entanglement?  As the  number of subsystems grows, the entanglement
detection  problem  becomes  more  and  more involved,  even  for  the  simplest
tripartite  case  (see  e.g.~\cite{Sve87,   Mer90,  SeeSve02}).  Indeed,  for  a
multipartite  system,  one  has   to  distinguish  between  the  so-called  full
separability and  many types  of partial separability  (see e.g.~\cite{HorHor09}
and  references therein).   Here  we  briefly discuss  a  possible extension  of
Proposition~\ref{prop:entcriteria} for  fully separable states.   The definition
of full multipartite  separability for $L$ subsystems acting  on a Hilbert space
$\mathcal{H}_L   =  \bigotimes_{l=1}^L   \H_l^{N_l}$  is   a   direct  extension
of~\eqref{eq:separable}, that is,
\begin{equation}\label{eq:separable multi}
\rho^{A_1\cdots A_L}_{\,\mathrm{FullSep}} = \sum_{m=1}^M \omega_m
\,\ketbra{\psi_m^{A_1}}{\psi_m^{A_1}} \otimes \cdots \otimes
\ketbra{\psi_m^{A_L}}{\psi_m^{A_L}}, \ \mbox{with} \ \omega_m \geq 0 \
\mbox{and} \ \sum_{m=1}^M \omega_m = 1.
\end{equation}

For multipartite fully separable states, we have the following:
\begin{proposition}
\label{prop:entcriteria multi}
Let  $\rho^{A_1\cdots A_L}_{\,\mathrm{FullSep}}$  be a  fully  separable density
operator of an  $L$-partite system, and let $\rho^{A_1},  \ldots, \rho^{A_L}$ be
the corresponding density operators of the subsystems. Then
\begin{equation}
\label{eq:entcriteria multi}
\SalicruQ(\rho^{A_1\cdots A_L}_{\,\mathrm{FullSep}}) \geq \max \left\{
\SalicruQ(\rho^{A_1}), \, \ldots , \, \SalicruQ(\rho^{A_L}) \right\},
\end{equation}
for any pair of entropic functionals $(h, \phi)$.
\end{proposition}

\begin{proof}
  The  majorization  relations~\eqref{eq:majorization  separable} for  separable
  bipartite states  are mainly  based on the  Schr\"odinger mixture  theorem, so
  that~\eqref{eq:majorization separable} can  be generalized to the multipartite
  case in  a direct  way~\cite{NieKem01}.  Let us  consider the  fully separable
  density  operator~\eqref{eq:separable multi},  written in  a diagonal  form as
  $\rho^{A_1\cdots     A_L}_{\,\mathrm{FullSep}}      =     \sum_i     \lambda_i
  \ketbra{e_i}{e_i}$. On  the one hand, from the  Schr\"odinger mixture theorem,
  there exists a unitary matrix $U$ such that
\begin{equation} \label{eq:mixth 1}
\sqrt{\lambda_i} \, \ket{e_i} = \sum_m U_{im} \sqrt{ \omega_m} \,
\ket{\psi_m^{A_1}} \otimes \cdots \otimes \ket{\psi_m^{A_l}} \otimes \cdots
\otimes \ket{\psi_m^{A_L}} .
\end{equation}
On     the    other     hand,    let     $\rho^{A_l}    =     \sum_m    \omega_m
\ketbra{\psi^{A_l}_m}{\psi^{A_l}_m}$  be  the  density  operator  of  the  $l$th
subsystem   and   its   diagonal   form   $\rho^{A_l}   =   \sum_j   \lambda^l_j
\ketbra{e^l_j}{e^l_j}$. Using again the  Schr\"odinger mixture theorem, there is
a unitary matrix $V^l$ such that
\begin{equation} \label{eq:mixth 2}
\sqrt{ \omega_m} \, \ket{\psi_m^{A_l}} = \sum_j V^l_{mj} \sqrt{\lambda_j^l} \,
\ket{e_j^l}.
\end{equation}
Substituting~\eqref{eq:mixth  2} into~\eqref{eq:mixth  1},  left-multiplying the
result     by     its      adjoint     and     using     the     orthonormality,
$\braket{e_j^l}{e_{j'}^l}=\delta_{jj'}$      and     $\braket{e_i}{e_{i'}}     =
\delta_{ii'}$, we obtain
\begin{equation} \label{eq:mixth 3}
\lambda = B^l \,\lambda^l,
\end{equation}
with  $B^l_{ij} = \sum_{m,m'}  U^*_{im'} U_{im}  \ V^{l  \, *}_{m'j}  V^l_{mj} \
\prod_{l'  \neq   l}  \braket{\psi_m^{A_{l'}}}{\psi_{m'}^{A_{l'}}}$.   Following
similar  arguments  as  in  the  bipartite  case~\cite[Th.~1]{NieKem01},  it  is
straightforward  to show  that  $B^l$ is  a  bistochastic matrix  and thus  that
$\lambda \prec  \lambda^l$.  Finally, using the  Schur-concavity of $\SalicruQ$,
we    obtain    $\SalicruQ(\rho^{A_1\cdots    A_L}_{\,\mathrm{FullSep}})    \geq
\SalicruQ(\rho^{A_l})$  for any  $l$, and  thus inequality~\eqref{eq:entcriteria
  multi}.
\end{proof}

Some  interesting  problems  to  be   addressed  are  the  application  of  this
proposition to particular  multipartite states, as well as  the extension of the
generalized entropic criteria to  different types of partial separability. These
points and related derivations are beyond the scope of the present contribution,
and will be the subject of future research.


\section{On possible further generalized informational quantities}
\label{s:Furtherdefinitions}

How  to obtain  useful conditional  entropies and  mutual informations  based on
generalized entropies is  an open question and there is  no general consensus to
answer it, even in the classical case (see e.g.~\cite{Fur06, Ras12:02, TeiMat12,
  Ras16} for different  proposals). Here, we first discuss  briefly two possible
definitions of classical conditional entropies and mutual informations, based on
$(h,\phi)$-entropies.   We then  proceed  to obtain  quantum  versions of  those
quantities.


\subsection{A  generalization  of  classical  conditional entropies  and  mutual
  informations}

Let us consider a pair of random variables $(A,B)$ with joint probability vector
$p^{AB}$, i.e.,\  $p^{AB}_{a,b} = \Pr[A=a,B=b]$, and  let us denote  by $p^A$ and
$p^B$  the corresponding  marginal  probability vectors,  namely $p^A_a=  \sum_b
p^{AB}_{a,b}$  and  $p^B_b  =   \sum_a  p^{AB}_{a,b}$.   From  Bayes  rule,  the
conditional probability vector for $A$ given that $B=b$, \ $p^{A|b}$, is defined
by $\displaystyle  p^{A|b}_a = \frac{p^{AB}_{a,b}}{p^B_b}$  (and analogously for
$p^{B|a}$). For the sake of convenience, in this section we indifferently denote
$\Salicru(A,B)$ or $\Salicru(p^{AB})$ (and similarly for the marginals).

In  order to  define  a conditional  $(h,\phi)$-entropy  of $A$  given $B$,  one
possibility is to take the  average of the $(h,\phi)$-entropy of the conditional
probability $p^{A|b}$ over all outcomes for $B$ (in a way similar to the Shannon
entropy~\cite{Sha48},  or also  to  the  R\'enyi and  Tsallis  entropies, as  in
Refs.~\cite{TeiMat12,Ras12:02}). This leads to the following:
\begin{definition}
\label{def:ClassicalConditionalJ}
Let us consider a pair of random variables $(A,B)$ with joint probability vector
$p^{AB}$.  We define the $\J$-conditional $(h,\phi)$-entropy of $A$ given $B$ as
\begin{equation}
\label{eq:ClassicalConditionalJ}
\Salicru^{\J}(A|B) \, = \, \sum_b p^B_b \: \Salicru(p^{A|b}).
\end{equation}
\end{definition}
From Def.~\ref{def:ClassicalConditionalJ},  one can thus  define a ``$\J$-mutual
information'' as
\begin{equation}
\label{eq:ClassicalJ}
\J_{(h,\phi)}(A;B) = \Salicru(A) - \Salicru^{\J}(A|B).
\end{equation}
However,   except  when   $h$   is   concave,  there   is   no  guarantee   that
$\mathcal{J}_{(h,\phi)}(A;B)$            is           nonnegative           (see
Proposition~\ref{prop:concavity}).  Besides,  this quantity is  not symmetric in
general.

An alternative proposal is to mimic the chain rule satisfied by Shannon entropy,
$H(A|B) = H(A,B) - H(B)  $, and thus to define a conditional entropy as follows:
\begin{definition}
\label{def:ClassicalConditionalI}
Let us consider a pair of random variables $(A,B)$ with joint probability vector
$p^{AB}$.  We define the $\I$-conditional $(h,\phi)$-entropy of $A$ given $B$ as
\begin{equation}
\label{eq:ClassicalConditionalI}
\Salicru^{\I}(A|B) \, = \, \Salicru(A,B) - \Salicru(B).
\end{equation}
\end{definition}

It   can  be   shown  that   this  quantity   is  nonnegative   from  Petrovi\'c
inequality~\cite[Th.~8.7.1]{Kuc09} together  with the appropriate  increasing or
decreasing    behavior   of~$h$    (see    also   the    properties   of    the
$(h,\phi)$-entropies, Section~\ref{s:Review}).

From   Def.~\ref{def:ClassicalConditionalI},   one   can  define   a   symmetric
``$\I$-mutual information'' as
\begin{equation}
\label{eq:ClassicalI}
\I_{(h,\phi)}(A;B) = \Salicru(A) - \Salicru^{\I}(A|B) = \Salicru(A) + \Salicru(B) -
\Salicru(A,B).
\end{equation}
Notice  that  $\J_{(h,\phi)}(A;B)$  and  $\I_{(h,\phi)}(A;B)$  coincide  in  the
Shannon case, i.e.,  for \ $h(x)=x, \phi(x)=-x\ln x$, but  they are different in
general.  Besides, like $\J_{(h,\phi)}(A;B)$,  $\I_{(h,\phi)}(A;B)$ can  also be
negative.

Other  alternatives  have  been  proposed  in specific  contexts,  such  as  for
Tsallis~\cite{Fur06,  Ras12:02,   Ras16}  or  R\'enyi  entropies~\cite{MulDup13,
  FerBer14, TomBer14,  TeiMat12}, but a unified  point of view  is still missing
and the question remains open.


\subsection{From  classical  to quantum  generalized  conditional entropies  and
  mutual informations}

Let $\rho^{AB}$ be the density operator  of a quantum composite system $AB$, and
let  $\Pi^B  =  \{\Pi^B_j\}$  be   a  local  projective  measurement  acting  on
$\H^{N_B}$,  i.e.,\ $\Pi^B_j  \Pi^B_{j'} =  \delta_{jj'}$  and $\sum_{j=1}^{N_B}
\Pi^B_j = I_{N_B}$.  The density  operator relative to that measurement is given
by $\displaystyle \rho^{A|\Pi^B} = \sum_j p_j \, \rho^{A|\Pi^B_j} $, where
$p_j =  \Tr(I \otimes \Pi^B_j \, \rho^{AB})$  and $\rho^{A|\Pi^B_j}= \frac{I
  \otimes \Pi^B_j  \, \rho^{AB} \,  I \otimes \Pi^B_j}{\Tr(I \otimes  \, \Pi^B_j
  \rho^{AB})}$.

We  define  the  quantum  version  of~\eqref{eq:ClassicalConditionalJ}  for  $A$
conditioned to the local projective measurement $\Pi^B$ acting on $B$, as
\begin{equation}
\label{eq:QuantumConditionalPiJ}
\SalicruQ^{\J}\left(A|B_{\Pi^B} \right) = \sum_j p_j \: \SalicruQ\big(
\rho^{A|\Pi^B_j} \big).
\end{equation}
This  quantity has  been proposed  for trace-form  entropies in~\cite{GigRos14}.
One can  obtain an  independent quantum conditional  entropy taking  the minimum
over       the       set       of       local       projective       measurement
in~\eqref{eq:QuantumConditionalPiJ}, which leads to the following:
\begin{definition}
\label{def:QuantumConditionalJ}
Let $\rho^{AB}$ be the density operator  of a quantum composite system $AB$, and
let  $\Pi^B$ be  a local  projective measurement  acting on  $B$. We  define the
quantum $\J$-conditional $(h,\phi)$-entropy of $A$ given $B$ as
\begin{equation}
\label{eq:QuantumConditionalJ}
\SalicruQ^{\J}( A|B ) = \min_{\{\Pi^B\}} \SalicruQ^{\J}\left(A|B_{\Pi^B} \right).
\end{equation}
\end{definition}

Now,     we     propose     the     quantum    version     of     the     mutual
information~\eqref{eq:ClassicalJ} as
\begin{equation}\label{eq:QuantumJ}
\JQ_{(h,\phi)}(A ; B) = \SalicruQ(\rho^A)-\SalicruQ^{\J}(A|B).
\end{equation}
Notice  that  if   $h$  is  concave,  then  $\JQ_{(h,\phi)}   \ge  0$,  but  its
nonnegativity       is      not       guaranteed      in       general      (see
Proposition~\ref{prop:concavity}).

Another  possibility  is  to  extend  the standard  definition  of  the  quantum
conditional entropy to $(h,\phi)$-entropies,  which leads to the quantum version
of~\eqref{eq:ClassicalConditionalI}:
\begin{definition}
\label{def:QuantumConditionalI}
Let $\rho^{AB}$ be the density operator  of a quantum composite system $AB$, and
$\rho^B$ be the  corresponding density operator of subsystem  $B$. We define the
quantum $\I$-conditional $(h,\phi)$-entropy of $A$ given $B$ as
\begin{equation}
\label{eq:QuantumConditionalI}
\SalicruQ^{\I}(A|B) = \SalicruQ(\rho^{AB}) - \SalicruQ(\rho^B).
\end{equation}
\end{definition}

Notice  that,  contrary  to  the   classical  case,  these  quantities  are  not
necessarily positive, except when $\rho^{AB}$ is separable, as we have precisely
shown in Proposition~\ref{prop:entcriteria}.

On the other  hand, one can propose a  quantum version of~\eqref{eq:ClassicalI},
as follows
\begin{equation}\label{eq:QuantumI}
\IQ_{(h,\phi)}(A;B) = \SalicruQ(\rho^A) + \SalicruQ(\rho^B) -
\SalicruQ(\rho^{AB}).
\end{equation}
Note however that there is no  guarantee of nonnegativity of these quantities in
general (see discussion before Proposition~\ref{prop:NoSuperadditivity}).

Unlike the  classical case, the  quantum mutual informations~\eqref{eq:QuantumJ}
and~\eqref{eq:QuantumI} are  different even for  the von Neumann case,  and this
difference   is    precisely   the   origin   for   the    notion   of   quantum
discord~\cite{OllZur01}. However, attempting to extend the definition of quantum
discord through  a direct  replacement of von  Neumann entropy by  a generalized
$(h,\phi)$-entropy fails in general (see e.g.~\cite{Jur13, BelPlas14}).

Unfortunately,  all  the quantities  given  in  this  section remain  as  formal
definitions   until  one   does  not   provide   a  complete   study  of   their
properties. This task lies beyond the scope of the present work and is currently
under   investigation~\cite{BosBel16}  (see~\cite{BerSes15}   for   a  different
approach of the development of quantum information measures).


\section{Concluding remarks}
\label{s:Conclusions}

We have proposed a quantum  version of the $(h,\phi)$-entropies first introduced
by    Salicr\'u    \textit{et    al.}     for    classical    systems.     Along
Sec.~\ref{s:QuantumEntropies}   we   have  presented   our   main  results   and
definitions.  Indeed,  after the definition of  quantum $(h,\phi)$-entropy given
in Eq.~\eqref{eq:QuantumSalicru}, we have  derived many basic properties related
to Schur concavity and majorization,  valid for any pair of entropic functionals
$h$ and $\phi$ with the proper continuity, monotonicity and concavity properties
established  in  Def.~\ref{def:QuantumSalicru}.   Next,  we have  discussed  the
properties of the $(h,\phi)$-entropies in  connection with the action of quantum
operations and measurements.  And later we  have extended our study for the case
of  composite systems,  focusing  on important  properties  like additivity  and
subadditivity, and discussing an application to entanglement detection. Besides,
in  Sec.~\ref{s:Furtherdefinitions} an  attempt  to deal  with generalized  conditional
entropies  and mutual informations  was presented,  although it  deserves deeper
development, which is beyond our present scope.

The first advantage  of our general approach  has to do with the  fact that many
properties of particular  examples of interest (such as  von Neumann and quantum
R\'enyi  and  Tsallis entropies),  can  be studied  from  the  perspective of  a
unifying formal framework,  which explains in an elegant  fashion many of their
common properties. In particular, our analysis reveals that majorization plays a
key role in explaining most of the common features of the more important quantum
entropic measures.

Remarkably  enough, we have  shown that  many physical  properties of  the known
quantum  entropies, such  as  preservation under  unitary  evolution, and  being
nondecreasing under  bistochastic quantum operations, hold in  the general case.
This is a signal indicating that our proposal yields new entropy functions which
may be of interest for physical purposes.  This kind of study can also be of use
in applications to  the description of quantum correlations, as  we have seen in
Sec.~\ref{s:CompositeSystemsII} for  the case of bipartite  entanglement; in the
case of multipartite systems, we have also discussed a possible extension of the
entropic entanglement  criterion, however restricted to  fully separable states.
Finally,  we mention that  the present  proposal may  also have  applications to
quantum information  processing and the problem of  quantum state determination,
because it  allows for a  more general systematization  of the study  of quantum
entropies.


\section*{Acknowledgments}

GMB, FH,  MP and PWL  acknowledge CONICET and  UNLP (Argentina), and MP  and PWL
also acknowledge SECyT-UNC (Argentina) for  financial support. SZ is grateful to
the University of Grenoble-Alpes (France) for the AGIR financial support.


\bibliographystyle{spphys}
\bibliography{Generalized_quantum_entropy}

\end{document}